\documentclass{lmcs}
\pdfoutput=1
\usepackage[utf8]{inputenc}

\usepackage{lastpage}
\lmcsdoi{21}{4}{19}
\lmcsheading{}{\pageref{LastPage}}{}{}%
{Nov.~01,~2024}{Nov.~05,~2025}{}

\usepackage[T1]{fontenc}
\usepackage{microtype}


\usepackage{amsmath,amsthm,amssymb,mathtools}

\usepackage{stmaryrd}
\usepackage{bussproofs}
\usepackage{xcolor-material}
\usepackage{caption}
\usepackage{subcaption}
\usepackage{thm-restate}
\usepackage{aligned-overset}
\usepackage{hyperref}
\usepackage{cleveref}

\usepackage{float}
\restylefloat{figure}

\usepackage{graphicx}
\newcommand{\syntaxPurble}[1]{\textcolor{MaterialPurple}{#1}}

\usepackage{tikz}
\tikzset{root/.style={draw, circle, minimum size=0.5cm}}
\tikzset{rhsRoot/.style={draw, minimum size=0.5cm}}
\tikzset{
    purple arrow/.style = {
      ->, %
      color=\syntaxPurble{} %
    }
}
\tikzset{
    redirect arrow/.style={
        ->, %
        color=GoogleBlue, %
        dashed
    }
}

\usepackage{float}
\restylefloat{figure}

\usepackage{bm}

\newcommand{\sortFont}[1]{\mathsf{#1}}
\newcommand{\varFont}[1]{#1}

\newcommand{\csPairFont}[1]{\mathcal{#1}}
\newcommand{\compFont}[1]{\mathsf{#1}}

\newcommand{\consFont}[1]{\mathsf{#1}}
\newcommand{\defFont}[1]{\mathsf{#1}}

\newcommand{\metaFont}[1]{\mathtt{#1}}

\newcommand{\nat}{\sortFont{nat}}

\newcommand{\bit}{\sortFont{bit}}
\newcommand{\word}{\sortFont{word}}
\newcommand{\nnat}{\sortFont{nnat}}

\newcommand{\cost}{\compFont{c}}
\newcommand{\size}{\compFont{s}}

\newcommand{\zero}{\consFont{0}}
\newcommand{\nil}{\consFont{[]}}
\newcommand{\suc}{\consFont{s}}

\newcommand{\bito}{\consFont{o}}
\newcommand{\biti}{\consFont{i}}
\newcommand{\bitb}{\consFont{b}}
\newcommand{\consInfix}{\consFont{::}}

\newcommand{\natz}{\consFont{o}}
\newcommand{\natsuc}{\consFont{n}}

\newcommand{\makelist}[1]{[#1]}
\newcommand{\listSep}{;}
\newcommand{\limitsize}[2]{|#1|_{#2}}

\newcommand{\add}{\defFont{add}}

\newcommand{\mult}{\defFont{mult}}

\newcommand{\helper}{\defFont{helper}}
\newcommand{\minus}{\defFont{minus}}
\newcommand{\funcProd}{\defFont{fnProd}}

\newcommand{\blank}{\mathtt{B}}

\newcommand{\asort}{\iota}
\newcommand{\bsort}{\kappa}
\newcommand{\csort}{\nu}
\newcommand{\atype}{\sigma}
\newcommand{\btype}{\tau}
\newcommand{\ctype}{\rho}

\newcommand{\afun}{\defFont{f}}
\newcommand{\bfun}{\defFont{g}}

\newcommand{\funF}{\defFont{F}}
\newcommand{\funG}{\defFont{G}}
\newcommand{\funS}{\defFont{S}}

\newcommand{\avar}{x}
\newcommand{\bvar}{y}
\newcommand{\cvar}{z}

\newcommand{\aFuncVar}{F}

\newcommand{\aterm}{s}
\newcommand{\bterm}{t}
\newcommand{\cterm}{u}

\newcommand{\encode}[1]{\underline{\mathsf{#1}}}

\newcommand{\aDom}{\csPairFont{F}}

\newcommand{\costInt}[1]{\mathcal{C}_{#1}}

\newcommand{\sizeInt}[1]{\mathcal{S}_{#1}}

\newcommand{\aTOFunc}{f}
\newcommand{\bffDom}{W}

\newcommand{\syntaxSig}{\mathbb{F}}
\newcommand{\signature}{\Sigma}

\newcommand{\var}{\mathbb{X}}
\newcommand{\terms}{\mathsf{T}(\syntaxSig,\var)}

\newcommand{\vars}[1]{\mathtt{vars}(#1)}

\newcommand{\sortset}{\mathbb{B}}
\newcommand{\simpletypeset}{{\mathbb{T}(\sortset)}}

\newcommand{\rules}{\mathcal{R}}
\newcommand{\Nat}{\mathbb{N}}

\newcommand{\trs}{{(\syntaxSig,\rules)}}
\newcommand{\polM}{P_M}
\newcommand{\transitions}{\mathcal{T}}

\newcommand{\Pos}{\metaFont{pos}}
\newcommand{\emptyPos}{\sharp}

\newcommand{\rulesArrow}{\to}
\newcommand{\arrzR}{\rulesArrow_\rules}
\newcommand{\arrz}{\rulesArrow}
\newcommand{\arrzT}{\mathrel{\rulesArrow_\rules^{+}}}
\newcommand{\arrzTR}{\mathrel{\rulesArrow_\rules^{*}}}

\newcommand{\sizeinterpret}[1]{{\llbracket{} #1 \rrbracket}^\size}
\newcommand{\costinterpret}[1]{{\llbracket{} #1 \rrbracket}^\cost}
\newcommand{\totalcost}[1]{\mathsf{cost}(#1)}
\newcommand{\totalcostprime}[1]{\mathsf{cost^{\star}}(#1)}

\newcommand{\funcinterpret}[1]{\mathcal{J}_{#1}}
\newcommand{\funcinterpretsize}[1]{\mathcal{J}_{#1}^\size}
\newcommand{\funcinterpretcost}[1]{\mathcal{J}_{#1}^\cost}

\newcommand{\costGt}{>}
\newcommand{\costGe}{\geq}
\newcommand{\sizeGe}{\sqsupseteq}

\newcommand{\typeVec}[1]{\vec{#1}}

\newcommand{\polySet}[1]{
    {\mathtt{Pol}^2_\Nat[#1]}
}

\newcommand{\fatlambda}{\lambda\!\!\!\lambda}

\newcommand{\app}{\,}

\newcommand{\hasType}{\mathrel{:}}

\NewDocumentCommand{\irc}{o}{
    \IfValueTF{#1}{
        \metaFont{irc}_{\rules_{#1}}
    }{
        \metaFont{irc}_{\rules}
    }
}
\newcommand{\fpTime}{\ensuremath{\mathtt{FP}}}
\newcommand{\bff}{\ensuremath{\mathtt{BFF}}}
\newcommand{\bffT}{{\bff_{2}}}
\newcommand{\sfSymbol}{\defFont{S}_f}
\newcommand{\rulesExtdvec}{\rules_{+\vec{f}}}
\newcommand{\rulesExtd}{\rules_{+f}}

\newcommand{\subtermR}{\mathrel{\unrhd}}

\newcommand{\asub}{\gamma}

\newcommand{\defName}[1]{\textbf{#1}}
\newcommand{\defTerm}[1]{\emph{#1}}

\newcommand{\ar}{\metaFont{typeOf}}
\newcommand{\lvl}[1]{\metaFont{ord}(#1)}

\newcommand{\arrtype}{\Rightarrow}
\newcommand{\arrfunc}{\longrightarrow}
\newcommand{\arrfuncwm}{\Longrightarrow}

\newcommand{\cs}{cost--size} %
\newcommand{\cstitle}{Cost--Size} %

\newcommand{\unary}[1]{\ulcorner\mathsf{#1}\urcorner}

\newcommand{\aGraph}{G}

\newcommand{\bGraph}{H}
\newcommand{\cGraph}{I}
\newcommand{\dGraph}{J}

\newcommand{\arrg}{\rightsquigarrow}

\newcommand{\aVert}{v}

\newcommand{\gLab}{\metaFont{label}}
\newcommand{\gSucc}{\metaFont{succ}}
\newcommand{\gApp}{@}
\newcommand{\gVar}{\bot}
\newcommand{\gRoot}{{\Lambda}}
\newcommand{\gVert}{V}
\newcommand{\gEmptyW}{\varepsilon}
\newcommand{\gHom}{\phi}

\newcommand{\gRedirect}{\mathbin{\gg}}

\newcommand{\toGraph}[1]{[#1]_{\mathbb{G}}}
\newcommand{\fromGraph}[1]{[#1]^{-1}_\mathbb{G}}

\newcommand{\restr}[2]{\mathsf{reach}(#1,#2)}

\newcommand{\aTTFunc}{\Psi}

\newcommand{\startState}{q_0}
\newcommand{\startStateSymb}{\mathtt{q_0}}

\newcommand{\leftM}{L}
\newcommand{\rightM}{R}

\newcommand{\cfgSep}{{\#}}

\newcommand{\aOTM}{M}

\newcommand{\aState}{q}

\newcommand{\cfgTrans}{\rightsquigarrow}

\newcommand{\tapeL}{\consFont{L}}
\newcommand{\tapeR}{\consFont{R}}
\newcommand{\tapeSplit}{\consFont{split}}
\newcommand{\leftTy}{\sortFont{left}}
\newcommand{\rightTy}{\sortFont{right}}
\newcommand{\tapeTy}{\sortFont{tape}}

\newcommand{\stateSymb}{\defFont{q}}
\newcommand{\aStateSymb}{\mathtt{q}}
\newcommand{\bStateSymb}{\mathtt{l}}
\newcommand{\querystate}{\mathtt{query}}
\newcommand{\answerstate}{\mathtt{answer}}
\newcommand{\finalstate}{\mathtt{end}}

\newcommand{\machinestep}{\defFont{step}}
\newcommand{\clean}{\defFont{clean}}
\newcommand{\execute}{\defFont{execute}}
\newcommand{\len}{\defFont{len}}
\newcommand{\limit}{\defFont{limit}}
\newcommand{\returnif}{\defFont{retif}}
\newcommand{\tryapply}{\defFont{tryapply}}
\newcommand{\maxsymb}{\defFont{max}}
\newcommand{\tryall}{\defFont{tryall}}
\newcommand{\extract}{\defFont{extract}}

\newcommand{\cfgTy}{\sortFont{config}}
\newcommand{\init}[1]{\mathsf{initial}(#1)}
\newcommand{\setnil}{\consFont{\emptyset}}
\newcommand{\setcons}{\consFont{scons}}
\newcommand{\setsort}{\sortFont{set}}

\newcommand{\tobin}{\defFont{toBin}}
\newcommand{\lengthOf}{\defFont{lengthOf}}
\newcommand{\compute}{\defFont{compute}}
\newcommand{\main}{\defFont{start}}
\newcommand{\accV}{\varFont{acc}}
\newcommand{\counterV}{\varFont{i}}
\newcommand{\binAddInfix}{\mathbin{\defFont{+_{B}}}}

\definecolor{shade}{HTML}{F4F4FF}%

\allowdisplaybreaks%

\theoremstyle{definition}
\newtheorem{theorem}[thm]{Theorem}
\newtheorem{lemma}[theorem]{Lemma}
\newtheorem{corollary}[theorem]{Corollary}
\newtheorem{definition}[theorem]{Definition}
\newtheorem{example}[theorem]{Example}
\newtheorem{remark}[theorem]{Remark}

\begin{document}

\title[A Characterization of \texorpdfstring{\( \bff_2 \)}{BFF} Through Tuple Interpretations]{A Characterization of Basic Feasible Functionals Through Higher-Order Rewriting and Tuple~Interpretations}

\thanks{This work is supported by the NWO TOP project
``Implicit Complexity through Higher-Order Rewriting'',
NWO 612.001.803/7571,
the NWO VIDI project
``Constrained Higher-Order Rewriting and Program Equivalence'',
NWO VI.Vidi.193.075,
and the ERC CoG
``Differential Program Semantics'',
GA 818616.
}

\author[P.~Baillot]{Patrick Baillot\lmcsorcid{0009-0002-9364-1140}}[a]
\author[U.~Dal~Lago]{Ugo Dal Lago\lmcsorcid{0000-0001-9200-070X}}[b]
\author[C.~Kop]{Cynthia Kop\lmcsorcid{0000-0002-6337-2544}}[c]
\author[D.~Vale]{Deivid Vale\lmcsorcid{0000-0003-1350-3478}}[c]

\address{Univ. Lille, CNRS, Inria, Centrale Lille, UMR 9189 CRIStAL, F-59000 Lille, France}
  \email{patrick.baillot@univ-lille.fr}

\address{Universit\`a di Bologna, Inria}
  \email{ugo.dallago@unibo.it}

\address{Radboud University Nijmegen}
  \email{c.kop@cs.ru.nl, deividvale@cs.ru.nl}

\begin{abstract}
  \noindent
  The class of type-two basic feasible functionals (\( \bffT \)) is the analogue of
  \fpTime{} (polynomial time functions) for type-2 functionals, that is,
  functionals that can
  take (first-order) functions as arguments\@. \( \bffT \) can be defined through
  Oracle Turing Machines with running time bounded by second-order polynomials.
  On the other hand, higher-order term rewriting
  provides an elegant formalism for expressing higher-order computation.
  We address the problem of characterizing \( \bffT \) by higher-order term rewriting.
  Various kinds of interpretations for \textit{first-order}
  term rewriting have been introduced in the literature for proving termination and characterizing first-order complexity classes.
  In this paper,
  we consider a recently introduced notion of \cs{} tuple interpretations
  for higher-order term rewriting and see
  second order rewriting as ways of computing type-2 functionals.
  We then prove that the class of functionals represented by higher-order terms
  admitting polynomially bounded \cs{} interpretations exactly corresponds to \( \bffT \).
\end{abstract}

\maketitle

\section{Introduction}\label{section:introduction}

Computational complexity classes --- and in particular those relating to
polynomial time and space~\cite{DBLP:journals/isci/HartmanisS69,cobham:65} ---
capture the concept of a feasible problem, and as
such have been scrutinized with great care by the scientific community in the
last fifty years.
The fact that even apparently simple problems
(such as nontrivial separation between those classes)
remain open today has highlighted the need for a comprehensive study aimed at investigating the deep nature of
computational complexity.
The so-called implicit computational
complexity~\cite{bellantoni:cook:92,leivant:91,oitavem:01,DBLP:journals/tcs/LagoH11,DBLP:journals/iandc/BaillotBR18}
fits into this picture, and is concerned with characterizations of complexity
classes based on tools from mathematical logic and the theory of programming
languages.

One of the areas involved in this investigation is certainly that of
term rewriting~\cite{terese}, which has %
proved useful as a tool for the
characterization of complexity classes. %
In particular, the class \fpTime{} (i.e., of first-order functions computable in polynomial time) has been
characterized through variations of techniques originally
introduced for \emph{termination}, e.g., the %
interpretation method~\cite{manna:ness:70,lankford:79},
path orders~\cite{dershowitz:82}, or dependency pairs~\cite{giesel:thiemann:schneider:04}.
Some examples of such characterizations can be found in~\cite{beckmann:weiermann:96,DBLP:journals/jfp/BonfanteCMT01,DBLP:journals/tcs/BonfanteMM11,DBLP:journals/corr/AvanziniM13,DBLP:conf/csl/BaillotL12}.

After the introduction of \fpTime{},
it became clear that the study of computational complexity
also applies to \textit{higher-order functionals},
which are functions that take not only data but also other functions as inputs.
The pioneering work of
Constable~\cite{constable:73},
Mehlhorn~\cite{mehlhorn:76},
and Kapron~and~Cook~\cite{kapron:cook:96}
laid the foundations of the so-called higher-order complexity,
which remains a prolific research area to this day.
Some motivations for this line of work can be found for instance
in computable analysis~\cite{DBLP:journals/toct/KawamuraC12},
NP search problems~\cite{DBLP:journals/jcss/BeameCEIP98},
and programming language theory~\cite{DBLP:conf/popl/DannerR06}.

There have been several proposals for a class of type-2
functionals that generalizes \fpTime{}.
However, the most widely accepted one is the class \( \bffT \)
of \emph{type-two basic feasible functionals}.
This class can be characterized based on function algebras,
similar to Cobham-style,
but it can also be described using Oracle Turing Machines.
The class \( \bffT \) was then the object of study by the research community,
which over the years has introduced
a variety of characterizations, e.g.,
in terms of programming languages with
restricted recursion schemes~\cite{DBLP:journals/jfp/IrwinRK01,DBLP:conf/popl/DannerR06},
typed imperative languages~\cite{DBLP:conf/lics/HainryKMP20,DBLP:conf/fossacs/HainryKMP22},
and restricted forms of iteration in OTMs~\cite{DBLP:conf/lics/KapronS18}.

The studies cited above present structurally complex programming languages and
logical systems, precisely due to the presence of higher-order functions.
It is not currently known whether it is possible to give a characterization of
\( \bffT \) in terms of mainstream concepts of rewriting theory,
although the latter has long been known to provide tools
for the modeling and analysis of functional programs with higher-order
functions~\cite{klop:oostrom:raamsdonk:93}.

This paper goes precisely in this direction by showing that the interpretation method in the form studied by Kop~and~Vale~\cite{DBLP:conf/fscd/KopV21,kop:vale:23} provides the right tools to characterize \( \bffT \).
More precisely,
we consider a class of higher-order rewrite systems admitting \cs{} tuple interpretations
(with some mild upper-bound conditions on their cost and size components)
and show that this class contains exactly the functionals in \( \bffT \).
Such a characterization could not have been obtained employing classical
integer interpretations as e.g.\@ in~\cite{DBLP:journals/jfp/BonfanteCMT01}
because \( \bffT \) crucially relies on some conditions both on \emph{size} and on \emph{time}.
We believe that a benefit of this characterization is that it opens the way to effectively handling programs
or executable specifications implementing \( \bffT \) functions, in full generality.
For instance,
we expect that such a characterization could be integrated into rewriting-based
tools for complexity analysis of term rewriting systems
such as e.g.,~\cite{DBLP:conf/tacas/AvanziniMS16}.

\paragraph{Contributions}
We consider simply-typed term rewriting systems (STRS for short)
with partial application but no \( \lambda \)-abstraction.
The contributions of this paper are as follows.
\begin{itemize}
  \item We provide a compatibility theorem for \cs{} interpretations of STRSs with respect to
    innermost reduction, whose proof is simpler than that of \cite{kop:vale:23}; this captures the fact
    that \cs{} interpretations provide safe upper bounds on the length of reduction sequences, and that the size interpretation of a term cannot increase.

  \item  We propose a natural definition of computation of type-2 functionals by an STRS, which can be of more general interest than the specific characterization of the class \( \bffT \) that we focus on in this paper;
  this is purely operational in nature and revolves around the use of rewrite rules modeling calls for an ``oracle''.
  This notion of computation is, intuitively, a rewriting counterpart for oracle Turing machines.

  \item We prove a soundness result, stating that any (orthogonal) STRS with a polynomially bounded \cs{} interpretation computes a type-2 functional in \( \bffT \); this proof uses, in particular, a term-graph rewriting argument.

  \item Conversely,
  we define an encoding of polynomial time Oracle Turing Machines in an STRS
  which shows that any type-2 functional in \( \bffT \) can be computed
  by an STRS with a polynomially bounded \cs{} interpretation.
\end{itemize}

\paragraph{Related Work}
We describe here some related work about the topics addressed by this paper,
namely implicit computational complexity,
higher-order complexity classes,
higher-order rewriting systems,
and interpretations.

Implicit computational complexity refers to a line of work
aiming at characterizing complexity classes without
reference to machine models and explicit bounds on resources,
but instead by relying on logical systems and programming language restrictions.
It goes back to early work by Leivant~\cite{leivant:91} and
Bellantoni~and~Cook~\cite{bellantoni:cook:92}
and has used various methods coming in particular from
recursion theory~\cite{DBLP:conf/csl/Marion94,oitavem:01},
programming language restrictions~\cite{DBLP:journals/jfp/Jones01,DBLP:journals/lmcs/KopS17},
linear logic~\cite{DBLP:journals/iandc/Girard98,DBLP:journals/tcs/Lafont04,DBLP:conf/aplas/Baillot11} and
 type systems~\cite{DBLP:journals/iandc/Hofmann03,DBLP:journals/iandc/LagoS16}.
In the setting of term-rewriting it has taken advantage of contributions
in the area of polynomial interpretations~\cite{DBLP:conf/csl/BonfanteCMT98}
and has provided a variety of characterizations for
first-order complexity classes such as
\fpTime\ and \( \mathtt{PSPACE} \)~\cite{%
DBLP:conf/lpar/MarionM00,%
DBLP:conf/rta/BonfanteMM05,%
beckmann:weiermann:96,%
DBLP:journals/jfp/BonfanteCMT01,%
DBLP:journals/tcs/BonfanteMM11,%
DBLP:journals/corr/AvanziniM13,%
DBLP:journals/mscs/BaillotLM12%
}.

The class of Basic Feasible Functionals \( \bff \) was introduced by
Cook~and~Kapron~\cite{DBLP:conf/focs/CookK89} by means of bounded typed loop programs,
and they showed that its type-2 restriction \( \bffT \) coincides with a class that
had been defined by Melhorn~\cite{mehlhorn:76}.
They later provided a machine characterization of type-2 \( \bffT \)
by polynomial time Oracle Turing Machines
{(OTM)}~\cite{DBLP:conf/focs/KapronC91,kapron:cook:96}
which gave more confidence in the naturalness of this class.
Several works then provided alternative characterizations of \( \bffT \),
in particular by restricted recursion schemes in some
functional languages~\cite{DBLP:journals/jfp/IrwinRK01,DBLP:conf/popl/DannerR06},
or typed imperative languages with insights coming from
non-interference
analysis~\cite{DBLP:conf/lics/HainryKMP20,DBLP:conf/fossacs/HainryKMP22},
or by restricted forms of iteration in OTMs~\cite{DBLP:conf/lics/KapronS18}.

Higher-order interpretations have been introduced and investigated
in~\cite{PhDvdPol} in relation to termination issues
but not to complexity classes.
In~\cite{DBLP:conf/csl/BaillotL12,DBLP:journals/iandc/BaillotL16}
a notion of higher-order polynomial interpretations was proposed
which allowed to provide a characterization of the (first-order)
class \( \fpTime \) of polynomial time computable functions.
However, the codomain considered was the domain \( \Nat \) of natural numbers,
not tuples, and this approach did not consider higher-order complexity.
An investigation of higher-order complexity classes employing the higher-order
interpretation method in the context of a pure higher-order functional language
was proposed in~\cite{DBLP:journals/lmcs/HainryP20}.
However, this paper did not provide a characterization of the standard \( \bffT \) class.
Instead, it characterized a newly proposed class \( \mathtt{SFF}_2 \)
(Safe Feasible Functionals) which is defined as the restriction of \( \bffT \)
to argument functions in \fpTime{}
(see Sect.~4.2 and the conclusion in~\cite{DBLP:journals/lmcs/HainryP20}).
Another related line of work is that of \cite{DBLP:journals/tcs/FereeHHP15}: the authors consider a first order
functional stream language, which allows to implement second order functionals and they study (first-order) interpretations for this language.
In this way they define a subclass  of  \( \bffT \) but do not characterize  \( \bffT \) itself.
Closer to the present work, the paper~\cite{DBLP:conf/fscd/KopV21} introduced the notion of
tuple interpretation for higher-order rewriting systems
and in particular the \cs{} interpretations that we use here.
Then~\cite{kop:vale:23}
adapted this notion to the weak call-by-value strategy,
which allowed for tighter bounds
but also relaxed the approach by separating the cost and the size components.

\paragraph{Publication History}
This paper is a revised and extended version of~\cite{baillot:lago:kop:vale:24}.
This version contains detailed proofs of the results and additional
examples, which previously could not be added due to space constraints.
The results in this paper appear in monograph form
in the last author's PhD thesis~\cite[Chapter 6]{val:24}.

\paragraph{Outline of the Paper}
  The paper is organized as follows.
  We first provide some background on higher-order rewriting and simply typed term-rewriting system (STRS), and on type-two complexity (\cref{sec:preliminaries}).   We then recall the definition of \cs{} interpretations and prove a compatibility theorem for \cs{} interpretations of STRSs with respect to innermost reduction (\cref{sec:cs-int}).
  In \cref{sec:rw-bff} we state the main theorem of this paper, which says that the STRSs with polynomially bounded \cs{} interpretations exactly characterize the \( \bffT \) functionals.
  \cref{sec:soundness} is devoted to the proof of the first direction of this theorem, the soundness result.
  \cref{sec:completeness} contains the proof of the second direction, the completeness result.
  In \cref{sec:conclusion} we conclude the paper and discuss future work.

\paragraph{Technical Overview.}
  In this paper we see \( \bffT \) (\cref{def:bff-class})
  as the set of those type-2 functionals
  computed by an Oracle Turing Machine in polynomial time.
  We recall basic definitions of such theory in \cref{subsec:bff}.
  The main result of this paper is a characterization of the
  class \( \bffT \) via higher-order term rewriting.

  In order to formally give a statement of this result
  we need to first establish some important notions
  such as the rewriting counterpart of an oracle
  (see \cref{def:rw-oracle-simulation})
  and a computability notion for higher-order rewriting
  (see \cref{def:type-2-rw-computability}).
  We state the main result in \cref{thm:main}.
  It is proved in two parts.
  We first prove that if any term rewriting system in this class computes
  a higher-order functional, then this functional has to be in \( \bffT \) (\emph{soundness}).
  Conversely,
  we prove that all functionals in \( \bffT \) are computed by this class of rewriting systems
  (\emph{completeness}).
  We argue that the key ingredient towards achieving this characterization
  is the ability to split the dual notions of \textit{cost} and \textit{size} given by
  the usage of tuple interpretations.

  Soundness at first seems straightforward.
  From Kop~and~Vale~\cite{kop:vale:23}
  we know that (call-by-value) higher-order rewrite systems that admit polynomial
  interpretations (with certain conditions on the interpretation of data constructors)
  satisfy the property that their runtime complexity is polynomially bounded.
  We could temptingly say if a term rewrite sytem computes
  (in the sense of \cref{def:type-2-rw-computability}) a type-2 functional,
  it must do so in a polynomial number of steps,
  and hence the said functional must be in \( \bffT \).
  However, this is not generally the case due to the size-explosion problem,
  i.e., in a polynomial number of steps we could iterate over copied data.
  We solve this issue by restricting the interpretation of data constructors
  in \cref{def:polynomially-bounded-int} and
  by employing term graph rewriting (see \cref{subsec:graphs}).
  Additionally, we need to guarantee that polynomial interpretations alone are
  capable of controlling the size of the calls to the oracle.
  More precisely, such interpretations do not allow for unbounded repeated iteration of oracle
  calls and the size of the resulting oracle call
  is polynomially related to the size of its given input.
  This is established by the \textit{Oracle Subterm Lemma} (see \cref{lemma:oracle-subterm-lemma}).
  We then prove soundness in \cref{thm:soundness}.

  To prove completeness we work on an encoding of polynomial time Oracle
  Turing Machines (OTMs) in STRSs\@.
  We proceed as follows:
  we encode machine configurations as terms and machine transitions as rewriting
  rules that rewrite such configuration terms.
  With such encoding,
  we can faithfully simulate transitions on an OTM as one or more rewriting steps
  on the corresponding rewrite system.
  The correctness of such simulation is the subject of \cref{lem:stepencode}.
  Notice however that simulating OTMs by STRSs is not enough.
  We need to do it in polynomially many steps.
  For this we provide a rewrite system that can fully simulate a run of an OTM
  in polynomially many steps, which is given by \cref{thm:completeness}.

\section{Preliminaries}\label{sec:preliminaries}

In this section
we present a brief overview of simply typed term rewriting
and the basic notions on basic feasible functionals.
We assume the reader to be familiar with concepts from rewriting theory~\cite{terese}
and basic notions of computability and complexity theory~\cite{AroraBarak}.

\subsection{Higher-Order Rewriting}\label{subsec:hotr}

We roughly follow the definition of a simply-typed term rewriting system (STRS)~\cite{kus:01}: types are simple, consisting of either base or functional types; terms are applicative, meaning they may be partially applied but contain no lambda abstractions.
We limit our interest to second-order STRSs where all rules have base type
(which in practice means the left- and right-hand side of rules are fully applied terms).
Reduction follows an innermost evaluation strategy.
We make this intuition precise below.

\smallskip
First, let us define our notion of a type.
Let \( \sortset \) be a nonempty set of \defTerm{base types}
which are ranged over by \( \asort, \bsort, \csort \).
The set \( \simpletypeset \) of \defTerm{simple types} over
\( \sortset \) is defined by the %
grammar
\( \simpletypeset \coloneq \sortset \mid \simpletypeset \arrtype \simpletypeset \).
Types from \( \simpletypeset \) are ranged over by \( \atype, \btype, \ctype \).
The \( \arrtype \) type constructor is right-associative,
so we write \( \atype \arrtype \btype \arrtype \ctype \)
for \( {( \atype \arrtype ( \btype \arrtype \ctype ) )} \).
Hence, every type \( \atype \) can be written as
\( \atype_1 \arrtype \cdots \arrtype \atype_n \arrtype \asort \) with $n \geq 0$.
We may write such types as \( \vec{\atype} \arrtype \asort \).
The \defTerm{order} of a type is: %
\( \lvl{\asort} = 0 \) for \( \asort \in \sortset \) and
\( \lvl{\atype \arrtype \btype} = \max(1 + \lvl{\atype}, \lvl{\btype}) \).

Then, signatures are sets of function symbols, each of them equipped with a type.
Formally, a \defTerm{signature} \( \syntaxSig \) is a triple \( {(\sortset, \signature, \ar)} \)
where \( \sortset \) is a set of base types,
\( \signature \) is a nonempty set of symbols,
and \( \ar \) is a function from \( \signature \) to \( \simpletypeset \).

Next, we define terms.
For each type \( \atype \) with \( \lvl{\atype} \leq 1 \),
we assume given a set \( \var_\atype \) of countably
many variables and assume that \( \var_\atype \cap \var_\btype = \emptyset \) if
\( \atype \neq \btype \).
We let \( \var \) denote \( \cup_\atype \var_\atype \)
and assume that \( \signature \cap \var = \emptyset \).
The set \( \terms \) --- of \defTerm{terms} built from \( \syntaxSig \) and \( \var \) ---
collects those expressions \( \aterm \) for which a judgment
\( \aterm \hasType \atype \) can be deduced using the following rules:
\begin{prooftree}
    \AxiomC{\( \avar \in \var_\atype  \)}
    \LeftLabel{\small{(ax)}}
    \UnaryInfC{\( \avar \hasType \atype \)}
    \DisplayProof{}
    \AxiomC{\( \afun \in \signature \vphantom{\avar \in \var_\atype}\)}
    \AxiomC{\( \ar(\afun) = \atype \vphantom{\avar \in \var_\atype}\)}
    \LeftLabel{\small{(f-ax)}}
    \BinaryInfC{\( \afun \hasType \atype \)}
    \DisplayProof{}
    \AxiomC{\( \aterm \hasType \atype \arrtype \btype \vphantom{\avar \in
      \var_\atype}\)}
    \AxiomC{\( \bterm \hasType \atype \vphantom{\avar \in \var_\atype}\)}
    \LeftLabel{\small{(app)}}
    \BinaryInfC{\( (\aterm \app \bterm) \hasType \btype \)}
\end{prooftree}
As usual, application of terms is left-associative,
so we write
\( \aterm \app \bterm \app \cterm \) for \( {((\aterm \app \bterm) \app \cterm)} \).
Let \(\vars{\aterm}\) be the set of variables occurring in \(\aterm \).
A term \( \aterm \) is \defTerm{ground} if \( \vars{\aterm} = \emptyset \).
The \defTerm{head symbol} of a term \( \afun \app \aterm_1 \cdots \aterm_n \)
is \( \afun \).
We say \( \bterm \) is a \defTerm{subterm} of \( \aterm \)
(written \( \aterm \subtermR \bterm \)) if either
(a) \( \aterm = \bterm \), or
(b) \( \aterm = \aterm' \app \aterm'' \) and \( \aterm' \subtermR
\bterm \) or \( \aterm'' \subtermR \bterm \).
It is a \defTerm{proper subterm} of \( \aterm \) if \( \aterm \neq \bterm \).
For a term \( s \), \( \Pos(s) \) is the set of \emph{positions} in \( s \):
\( \Pos(\avar) = \Pos(\afun) = \{ \emptyPos \} \) and \( \Pos(
s \app t) = \{ \emptyPos \} \cup \{ 1 \cdot u \mid u \in \Pos(s) \} \cup \{ 2 \cdot
u \mid u \in \Pos(t) \} \).
For \( p \in \Pos(s)\),
the subterm \( s|_p \) at position \(p\) is given by:
\( s|_\emptyPos = s \) and \( (s_1 \app s_2)|_{i \cdot p} = s_i|_p \).

In this paper, we require that for all \( \afun \in \signature \),
\( \lvl{\ar(\afun)} \leq 2 \), so w.l.o.g.,
\( \afun \hasType
  (\typeVec{\asort}_1 \arrtype \bsort_1)
  \arrtype \cdots \arrtype
  (\typeVec{\asort}_k \arrtype \bsort_k)
  \arrtype \csort_1 \arrtype \cdots \arrtype \csort_l \arrtype
  \asort
\).
Hence, in a fully applied term %
\( \afun \app \aterm_1 \ldots \aterm_k \app \bterm_1 \ldots \bterm_l \)
we say the \( \aterm_i \) are the arguments of type-1 and the \( \bterm_j \)
are the arguments of type-0 for \( \afun \).
We say a term \emph{has type order $k$} if its type has order $k$.
A \emph{substitution} \( \asub \) is a type-preserving map
from variables to terms
such that %
\({\{ \avar \in \var \mid \asub(\avar) \neq \avar \}}\)
is finite.
We %
extend \( \asub \) to %
terms as usual:
\( \avar \asub = \asub(\avar) \),
\( \afun \asub = \afun \), and
\( (\aterm \app \bterm) \asub = (\aterm \asub) \app (\bterm \asub) \).
A \defTerm{context} \( C \) is a term with a single occurrence of a variable \( \Box \);
the term \( C[s] \) is obtained by substituting \( \Box \) by \( s \).

Finally,
we have all the ingredients needed to define rewriting rules and the dynamics
they give rise to.
A \emph{rewrite rule} \( \ell \to r \) is a pair of terms of the same type such
that \( \ell \) has a form \( \afun \app \ell_1 \cdots \ell_m \) with
\( \afun \in \signature \), and \( \vars{\ell} \supseteq \vars{r} \).
It is \emph{left-linear} if no variable occurs more than once in \( \ell \).
A \emph{simply-typed term rewriting system} \( \trs \)
is a pair consisting of a signature \( \syntaxSig \) and
a set of rewrite rules \( \rules \) over \( \terms \).
In this paper, we require that all rules have \emph{base} type; that is,
for each rule \( \ell \to r \),
the type of \( \ell \) (so also the type of \( r \)) is a base type.
An STRS is \emph{innermost orthogonal} if all rules are left-linear,
and any two distinct rules \( \ell_1 \arrz r_1,\ell_2 \arrz r_2 \) are
non-overlapping:
there are no substitutions \( \gamma \) and \( \delta \) such that \( \ell_1\gamma = \ell_2\delta \).
A \emph{reducible expression} (redex) is a term of the form \( \ell \asub \) for
a rule \(\ell \rulesArrow r\) and substitution \( \asub \).
A term is in \defTerm{normal form} if none of its subterms is a redex.

The \textit{innermost rewrite relation} induced by \( \rules \) is defined as
follows:
\begin{itemize}
    \item \( \ell\gamma \arrzR r\gamma \), if \( \ell \arrz r \in \rules \) and
        \( \ell\gamma \) has no proper subterm that is a redex;
    \item \( \aterm \app \bterm \arrzR \cterm \app \bterm \),
        if \( \aterm \arrzR \cterm \) and
        \( \aterm \app \bterm \arrzR \aterm \app \cterm \),
        if \( \bterm \arrzR \cterm \).
\end{itemize}
We write \( \arrzT \) for the transitive closure of \( \arrzR \)
and \( \arrzTR \) for the transitive-reflexive closure.
An STRS \( \rules \) is \emph{innermost terminating} if no infinite rewrite
sequence \( \aterm \arrzR \bterm \arrzR \ldots \) exists.
It is \emph{innermost confluent} if \( \aterm \arrzT \bterm \) and \( \aterm \arrzT \cterm \)
implies that some \( v \) exists with \( \bterm \arrzTR v \) and \( \cterm \arrzTR v \).
It is well-known that innermost orthogonality implies innermost confluence~\cite{mayr:nipkow:98}.
In this paper, we will typically drop the ``innermost'' adjective and simply
refer to terminating/orthogonal/confluent STRSs.

\begin{example}\label{ex:add}\label{ex:times}\label{ex:rec}
Let us consider a simple example of STRS\@.
We start by defining the following signature \( \syntaxSig \).
The base types are \( \sortset = \{ \nat \} \) and \( \signature \) includes
\(
  \zero \hasType \nat,
  \suc \hasType \nat \arrtype \nat
\).
We add rules for basic operations such as
\( \add,\mult \hasType \nat \arrtype \nat \arrtype \nat \)
and a higher-order function
\( \funcProd \hasType (\nat \arrtype \nat) \arrtype \nat \arrtype \nat \arrtype \nat \).
Finally, the set of rules \( \rules \) is given by:
\[
\begin{array}{rclcrcl}
\add \app \zero \app y & \arrz & y & \quad &
\add \app (\suc \app x) \app y & \arrz & \suc \app (\add \app x \app y) \\
\mult \app \zero \app y & \arrz & \zero & \quad &
\mult \app (\suc \app x) \app y & \arrz & \add \app y \app (\mult \app x \app y)
  \\
\funcProd \app F \app \zero \app y & \arrz & y & \quad &
\funcProd \app F \app (\suc \app x) \app y & \arrz & \funcProd \app F \app x \app
  (\mult \app y \app (F \app x)) \\
\end{array}
\]
Notice that this system is left-linear and non-overlapping;
hence, it is innermost orthogonal, and therefore innermost confluent.

\end{example}

Hereafter,
we write \( \unary{n} \) for
the term
\( \suc \app (\suc \app (\dots \app \zero \dots)) \) with \( n \) \( \suc \)s.

\begin{remark}
In this paper, we have chosen to limit interest to second-order systems with base-type rules.
The restriction to base-type rules is not so significant --- since one may simply pad both sides of
the rules with fresh variables until they have base type --- but the order limitation is arguably
more restrictive.

We impose this restriction because our primary goal is to identify a convenient class of term
rewriting systems that characterizes \( \bffT \) --- \emph{not} to find the most general class of
term rewriting systems that may do so.
Indeed, we have already derived a variation of the results in this work for call-by-value systems
of unrestricted type order following the definitions of \cs{} interpretation in~\cite{kop:vale:23}
(see~\cite[Chapter 6]{val:24}).
Hence, we have every reason to believe that the type restriction could be dropped also in the
current innermost setting, provided the ``main'' function \( \funF \) that we will identify
in \cref{def:type-2-rw-computability} remains second-order, i.e., the oracle functions considered
are still first-order in such an envisioned extension.
We have not done so in this work because extending the definition of ``interpretation'' to higher-order innermost systems would add a lot of technical burden (as it also does for call-by-value reduction), and inhibit understanding of the methodology, for arguably little expressivity gain.

We also \emph{suspect} that the results will extend if lambda-abstraction and beta reduction are
permitted, but we have not confirmed this as it would require a larger change, primarily in the
construction of \cref{subsec:graphs}.  This is also not a major restriction, since abstractions in
the right-hand side of rules can be replaced by fresh function symbols.
\end{remark}

\subsection{Basic Feasible Functionals}\label{subsec:bff}

We assume familiarity with Turing Machines and basic notions of computability
and complexity theory~\cite{AroraBarak}. %
In this paper,
we consider \textit{deterministic multi-tape Turing Machines}.
Those are, conceptually, machines consisting of a finite set of states,
one or more (but a fixed number of) right-infinite \textit{tapes} divided into cells.
Each tape is equipped with a tape head that scans the symbols on the tape's cells and
may write on it.
The head can move to the left or right. A k-ary \textit{Oracle Turing Machine} (OTM)
is a deterministic multi-tape Turing Machine with at least \( 2k + 1 \) tapes:
one main tape for (input/output),
\( k \) designated \emph{query} tapes ($\mathtt{querytape}_i$ for $1 \leq i \leq k$)
and \( k \) designated \emph{answer} tapes ($\mathtt{answertape}_i$).
It also has \( k \) distinct \emph{query states} \( q_i \) and
\( k \) \emph{answer states} \( a_i \).

Let \( \bffDom = {\{0,1\}}^* \). A computation with a k-ary OTM \( \aOTM \) requires \( k \) fixed \textit{oracle functions}
\( \aTOFunc_1, \ldots, \aTOFunc_k : \bffDom \arrfunc \bffDom \).
We write \( \aOTM_{\vec{f}} \) to denote a run of \( \aOTM \) with these functions.
A run of \( \aOTM_{\vec{f}} \) on \( w \) starts with \( w \) written on the main tape.
It ends when the machine halts,
and yields the word that is written on the main tape as output.
As usual, we only consider machines that halt on all inputs.
The computation proceeds as usual for non-query states.
To query the value of \( f_i \) on \( w \),
the machine writes \( w \) on the corresponding query tape $\mathtt{tape}_{2*i}$
and enters the query state \( q_i \).
Then, \textit{in one step}, it transitions to the answer state \( a_i \) as follows:
\begin{enumerate}[label=(\alph*)]
\item the query value \( w \) written in $\mathtt{querytape}_i$, the query tape for \( \aTOFunc_i \), is read;
\item the contents of $\mathtt{answertape}_i$, are changed to \( \aTOFunc_i(w) \);
\item the query value \( w \) is erased from $\mathtt{querytape}_i$; and
\item the head of $\mathtt{answertape}_i$ is moved to its first symbol.
\end{enumerate}
The \defTerm{running time} of \( \aOTM_{\vec{f}} \) on \( w \)
is the number of steps used in the computation.

We say \( x \) is of type-0 whenever \( x \in \bffDom \).
A \defTerm{type-1 function} is a mapping in \( \bffDom \arrfunc \bffDom \).
A \defTerm{type-2 functional} of rank \( (k,l) \) is a mapping in
\( {(\bffDom \arrfunc \bffDom)}^k \arrfunc \bffDom^l \arrfunc \bffDom \).

\begin{definition}\label{def:type-2-computability}
We say an OTM %
\( \aOTM \)
\defName{computes} a type-2 functional \( \aTTFunc \) of rank \( (k,l) \)
iff for all
type-1 functions
\( \aTOFunc_1, \dots, \aTOFunc_k\)
and
\( x_1,\dots,x_l \in \bffDom \),
whenever \( \aOTM_{\aTOFunc_1,\dots,\aTOFunc_k} \) is started with
\( x_1, \dots,x_l \) written on its main tape (separated by blanks),
it halts with
\( \aTTFunc(\aTOFunc_1,\dots,\aTOFunc_k,x_1,\dots,x_l) \)
written on its main tape.
\end{definition}

\begin{definition}\label{def:poly-expr}
    Let
    \( {\{ F_1, \dots, F_k \}} \)
    be a set of \emph{type-1 variables} and
    \( {\{ x_1, \dots, x_l \}} \)
    a set of \emph{type-0 variables}.
    The set
    \( \polySet{F_1, \dots, F_k; x_1, \dots, x_l} \)
    of \defName{second-order polynomials} over \( \Nat \)
    with indeterminates \( F_1, \dots, F_k, x_1, \dots, x_l \)
    is generated by: %
    \[ P, Q \coloneq n \mid x \mid P + Q \mid P * Q \mid F(Q) \]
    where
    \( n \in \Nat \),
    \( x \in \{ x_1 , \dots, x_l \} \),
    and \( F \in \{ F_1, \dots, F_k \} \).
\end{definition}
Notice that such polynomial expressions can be naturally viewed as type-2 functionals, e.g.,
\( {P(F,x)} = {3 * F(x) + x} \) is a type-2 functional of rank \( (1,1) \).

Given \( w \in \bffDom \),
we write \( |w| \) for its length and define the length
\( |f| \) of \( f : \bffDom \arrfunc \bffDom \) as
\( |f| = \fatlambda n.
\max\limits
_{{|y|} \leq n} |f(y)|
\).
This allows us to define \( \bffT \) as the class of functionals computable by
OTMs with running time bounded by a second-order polynomial:

\begin{definition}\label{def:bff-class}
A type-2 functional \( \aTTFunc \) is in \( \bffT \)
iff
there exist an OTM \( \aOTM \) and a second-order polynomial \( P \)
such that
\( M \) computes \( \aTTFunc \) and for all \( \vec{\aTOFunc} \) and \( \vec{x} \):
the running time of \( \aOTM_{\aTOFunc_1,\dots, \aTOFunc_k} \) on \( x_1,\dots,x_l \)
is at most \( P(|\aTOFunc_1|,\dots, |\aTOFunc_k|,|x_1|,\dots,|x_l|) \).
\end{definition}

\section{\cstitle{} Interpretations}\label{sec:cs-int}

In this section
we define \cs{} interpretations for the syntax of
types and terms in the STRS format.
The interpretations we develop differ from those
in~\cite{kop:vale:23} in two key aspects:
here we limit ourselves to second-order systems and
the rewriting strategy is innermost instead of call-by-value.
As a consequence,
the notion of \textit{cost} here is explicitly expressed as a function
\( \totalcost{\cdot} \) that inductively computes the total cost of reduction.
Furthermore,
these changes require that we prove a higher-order version of the innermost
compatibility theorem.

\subsection{The Interpretation of Types and Terms}\label{subsec:interpretationdef}

In \cs{} interpretations, each base-type term is mapped to two different components:
a \emph{size} and a \emph{cost}.  Intuitively, the size component imposes a restriction
on the shape of the normal form of a term; for example, by limiting the total number of
function symbols in the normal form, or a weighted total where some symbols are weighed
more heavily than others.
Meanwhile, the cost component is an integer that bounds the number of steps needed to
reach a normal form.
By separating these two notions, we can obtain a more fine-grained complexity bound
compared to the classical interpretation methods that map to for instance natural
numbers~\cite{DBLP:journals/iandc/BaillotL16,DBLP:journals/jfp/BonfanteCMT01,DBLP:journals/lmcs/HainryP20}.
For non-base type terms, the cost and size interpretations are represented by
\emph{functions}.

In this paper both the cost and size components must be elements of quasi-ordered sets,
which we define formally as follows.

\paragraph{Size interpretation}
For sets \( A \) and \( B \),
we write \( A \arrfunc B \) for the set of functions from \( A \) to \( B \).
A \emph{quasi-ordered set} \( (A,\sizeGe) \) consists of a nonempty set \( A \)
and a reflexive and transitive relation \(\sizeGe \) on \( A \).
For quasi-ordered sets \( (A_1,\sizeGe_1) \) and \( (A_2,\sizeGe_2) \),
we write \( A_1 \arrfuncwm A_2 \) for the set of functions
\( f \in A_1 \arrfunc A_2 \) such that \( f(x) \sizeGe_2 f(y) \) whenever
\( x \sizeGe_1 y \), i.e.,
\(A_1 \arrfuncwm A_2 \) is the space of \emph{weakly monotonic} functions:
functions that preserve the quasi-ordering.

For every \( \asort \in \sortset \), let a quasi-ordered set
\( (\sizeInt{\asort},\sizeGe_\asort) \) be given.
We extend this to \( \simpletypeset \) by defining
\( \sizeInt{\atype \arrtype \btype} =
(\sizeInt{\atype} \arrfuncwm \sizeInt{\btype}, \sizeGe_{\atype \arrtype \btype})\)
where \( f \sizeGe_{\atype \arrtype \btype} g \)
iff \( f(x) \sizeGe_\btype g(x) \) for all \( x \in \sizeInt{\atype} \).
So \( \sizeInt{\atype \arrtype \btype} \) is a set of weakly monotonic functionals,
using pointwise comparison.

We assume given a fixed \emph{size interpretation function} \( \funcinterpretsize{} \),
which maps each \( \afun \in \signature \) to some \( \funcinterpretsize{\afun} \in
\sizeInt{\ar(\afun)} \).
If we are additionally given a valuation \( \alpha \) that maps each
\( x \in \var_\atype \) to \( \sizeInt{\atype} \), we can map each term
\( s \hasType \atype \) to an element of \( \sizeInt{\atype} \)
naturally as follows:
(a) \( \sizeinterpret{\avar}_\alpha = \alpha(\avar) \);
(b) \( \sizeinterpret{\afun}_\alpha = \funcinterpretsize{\afun} \);
(c) \( \sizeinterpret{s \app t}_\alpha = \sizeinterpret{s}_\alpha(\sizeinterpret{t}_\alpha) \).

(Note that the superscript ${}^\size$ always indicates \emph{size}, while we will use
${}^\cost$ for \emph{cost}.)

\paragraph{Cost interpretation}
For every type \( \atype \) with \( \lvl{\atype} \leq 2 \), we define
\( \costInt{\atype} \) as follows:
\begin{enumerate}[label=(\alph*)]
\item \( \costInt{\bsort} = \Nat \) for \( \bsort \in \sortset \);
\item \( \costInt{\asort \arrtype \btype} = \sizeInt{\asort} \arrfuncwm
  \costInt{\btype} \) for \( \asort \in \sortset \), and
\item \( \costInt{\atype \arrtype \btype} = \costInt{\atype} \arrfuncwm
  \sizeInt{\atype} \arrfuncwm \costInt{\btype} \) if \( \lvl{\atype} = 1 \).
\end{enumerate}
For the weak monotonicity, we again use pointwise comparison: for \( f,g \in
\costInt{\atype \arrtype \btype} \) we say \( f \costGe g \) if \( f(x) \costGe g(x) \)
for every \( x \in \costInt{\atype} \).

We assume given a fixed \emph{cost interpretation function}
\( \funcinterpretcost{} \), which maps each \( \afun \in \signature \) to some
\( \funcinterpretcost{\afun} \in \costInt{\ar(\afun)} \);
and for each type \( \atype \) with \( \lvl{\atype} \leq 1 \), we assume given
\defTerm{valuations} \( \alpha : \var_\atype \arrfunc \sizeInt{\atype} \)
and \( \zeta : \var_\atype \arrfunc \costInt{\atype} \).
We then define, for terms of type order $\leq 1$:
\[
\begin{array}{rcl}
\costinterpret{\avar \app s_1 \cdots s_n}_{\alpha,\zeta} & = &
  \zeta(\avar)(\sizeinterpret{s_1}_\alpha,\dots,\sizeinterpret{s_n}_\alpha) \\
\costinterpret{\afun \app s_1 \cdots s_k \app t_1 \cdots t_n}_{\alpha,\zeta} & = &
  \funcinterpretcost{\afun}(
    \costinterpret{s_1}_{\alpha,\zeta},\sizeinterpret{s_1}_\alpha, \ldots,
    \costinterpret{s_k}_{\alpha,\zeta},\sizeinterpret{s_k}_\alpha,
    \sizeinterpret{t_1}_\alpha,\ldots,
    \sizeinterpret{t_n}_\alpha
  ) \\
\end{array}
\]
This is well-defined under our assumptions that all variables have a type of order
\( 0 \) or \( 1 \) (so the arguments to a variable have type order $0$), and that we
can write
\( \afun \hasType
  (\typeVec{\asort_1} \arrtype \bsort_1)
  \arrtype
  \cdots
  \arrtype
  (\typeVec{\asort_k} \arrtype \bsort_k)
  \arrtype
  \csort_1 \arrtype \cdots \arrtype \csort_l \arrtype \asort
\) for each function symbol \( \afun \).
Observe also that \( \funcinterpretcost{\afun} \), when fully applied, maps
to a natural number.  As the left- and right-hand sides of rules have base type,
this definition assigns a natural number as the cost interpretation for both sides of
each rule.

Intuitively, for a ground base-type term \( \afun\ s_1 \cdots s_k\ t_1 \cdots t_n \),
\( \costinterpret{\afun \app s_1 \cdots s_k \app t_1 \cdots t_n}_{\alpha,\zeta} \)
expresses the cost of the reduction to normal form \emph{not including} the cost of
first normalising each $s_i$ and $t_i$.  To compute the full reduction cost, we let

\[ \totalcost{s}_{\alpha, \zeta} = \sum \{ \costinterpret{t}_{\alpha,\zeta} \mid
s \subtermR t \text{ and \( t \) is a non-variable term of base type}\} \]

In addition, in our proofs it is very convenient to have a second cost measure:
\[
  \totalcostprime{s}_{\alpha, \zeta} =
  \sum \{
    \costinterpret{t}_{\alpha,\zeta} \mid s \subtermR t
    \text{ and } t
    \text{ is a non-variable base-type term not in normal form}
  \}
\]
The difference between \( \mathsf{cost} \) and \( \mathsf{cost^{\star}} \) is that subterms
that are already in normal form may still contribute a non-zero part to the
\( \mathsf{cost} \) measure of a term, but not to \( \mathsf{cost^{\star}} \).

\paragraph{Compatibility}

A \defTerm{\cs{} interpretation} \( \aDom \) for a second order signature
\( \syntaxSig = (\sortset, \signature, \ar) \)
is a choice of a quasi-ordered set \( \sizeInt{\asort} \),
for each \( \asort \in \sortset \),
along with cost- and size-interpretation functions
\( \funcinterpretcost{} \) and \( \funcinterpretsize{} \) defined as above.
Let \( \trs \) be an STRS over \( \syntaxSig \).

\begin{definition}\label{def:compatible}
We say \( \trs \) is \defTerm{compatible} with a \cs{} interpretation
if for any valuations \( \alpha \) and \( \zeta \), we have
{(a)} \( \costinterpret{\ell}_{\alpha, \zeta} > \totalcost{r}_{\alpha,\zeta} \)
and
{(b)} \( \sizeinterpret{\ell}_{\alpha} \sizeGe \sizeinterpret{r}_{\alpha} \),
for all rules \( \ell \arrz r \) in \( \rules \).
In this case we say such \cs{} interpretation \defTerm{orients} all rules
in \( \rules \).
\end{definition}
As in all notions of interpretation where the underlying domain is that of the natural numbers, the compatibility check is, in general, bound to be undecidable.
This is the case even if sound but complete techniques exist.
In any case, a study about the possibility of automating our technique is outside the scope of this paper and is thus left to future work.

To avoid heavy notation, we will often omit the valuations \( \alpha \) and \( \zeta \)
from subscripts when quantifying over them, and keep the quantification implicit.
For example, instead of writing
\[
\forall \alpha,\zeta.\
  \costinterpret{\funcProd \app F \app (\suc \app x) \app y}_{\alpha, \zeta} >
  \totalcost{\funcProd \app F \app x \app (\mult \app y \app (F \app x))}_{\alpha,\zeta}
\]
we will use:
\[
  \costinterpret{\funcProd \app F \app (\suc \app x) \app y} >
  \totalcost{\funcProd \app F \app x \app (\mult \app y \app (F \app x))}
\]
When we do this, we will also leave \( \alpha \) and \( \zeta \) out of concreate
interpretations; for example, instead of
\[
  \costinterpret{\funcProd \app F \app x \app y}_{\alpha, \zeta} =
  3 * \alpha(x) * \alpha(y) * {\max(\alpha(F)(\alpha(x)),1)}^{\alpha(x)+1} + \alpha(x) * \zeta(F)(\alpha(x)) + 2 * \alpha(x) + 1
\]
we will use:
\[
  \costinterpret{\funcProd \app F \app x \app y}_{\alpha, \zeta} =
  3 * x * y * {\max(F^\size(x),1)}^{x+1} + x * F^\cost(x) + 2 * x + 1
\]
Where \( \alpha(x) \) is simply denoted as \( x \) for type-0 variables, and as
\( x^\size \) for type-1 variables; \( \zeta(x) \) is always denoted as
\( x^\cost \) (note that, when considering the cost interpretation of rules,
\( \zeta(x) \) only occurs for type-1 variables).

\begin{example}\label{ex:addinterpret}
  Let us consider the signature of the STRS given in \cref{ex:add}.
  We can interpret the base type \( \nat \) as the natural quasi-order set
  \( \sizeInt{\nat} = (\Nat,\geq) \).
  We need to give interpretations for the constructors of \( \nat \).
  Let us start by size.
  \begin{align*}
    \funcinterpretsize{\zero} &= 0 & \funcinterpretsize{\suc} &= \fatlambda x.x+1
  \end{align*}
  This gives us \( \sizeinterpret{\unary{n}} = n \) for all \( n \in \Nat \).
  So intuitively our interpretation is counting how many \( \suc \) symbols there
  are in the unary representation of a natural number, and setting this as its size.
  For the cost measure we do as follows.
  \begin{align*}
    \funcinterpretcost{\zero} &= 0 &  \funcinterpretcost{\suc} &= \fatlambda x.0
  \end{align*}
  This gives us \( \costinterpret{\unary{n}} = \totalcost{n} = 0 \),
  which makes sense since the base type \( \nat \)
  is intended to represent data values,
  from which no computation can be started.
  Here we see a first contrast with the classical interpretation method.
  By splitting cost and size we can faithfully encode
  complexity information into the interpretation.

  Next, we give size and cost interpretation functions for \( \add \) and \( \mult \).
  So we let \( \funcinterpretsize{\add} = \fatlambda xy.x+y \) and
  \( \funcinterpretsize{\mult} = \fatlambda xy.x*y \); then indeed
  \( \sizeinterpret{\ell} \geq \sizeinterpret{r} \) for the first four rules of
  \cref{ex:add}
  (e.g., \( \sizeinterpret{\mult \app (\suc \app x) \app y}
  = (x+1) * y \geq y + (x * y) = \sizeinterpret{\add \app y \app (\mult \app x
  \app y)} \)).
  For the cost interpretation functions, note that both \( \add \) and \( \mult \)
  are type-1, so we must choose functions that take the \emph{sizes} of their
  arguments as input, and return a cost.
  Here, we choose \( \funcinterpretcost{\add} = \fatlambda x y.x + 1 \) and
  \( \funcinterpretcost{\mult} = \fatlambda x y.x * y + 2 * x + 1 \).
  To see why this makes sense, consider arguments in \emph{normal form}:
  e.g., reducing \( \add\ (\suc^n\ \zero)\ (\suc^m\ \zero) \) to normal form takes
  \( n + 1 \) steps.  We indeed have the required inequality
  \( \costinterpret{\ell} > \totalcost{r} \) for the first four rules; for
  example:
  \[
  \begin{array}{rclcl}
  \costinterpret{\mult \app (\suc \app x) \app y} & = &
  \sizeinterpret{\suc \app x} * \sizeinterpret{y} + 2 * \sizeinterpret{\suc \app x} + 1 \\
  & = & (x + 1) * y + 2 * (x + 1) + 1 & = & x * y + 2 * x + y + 3 \\
  & > \\
  \totalcost{\add \app y \app (\mult \app x \app y)} & = &
  \costinterpret{\add \app y \app (\mult \app x \app y)} + \costinterpret{\mult \app x \app y} \\
  & = & (y+1) \quad + \quad (x * y + 2 * x + 1) & = & x * y + x * x + y + 2 \\
  \end{array}
  \]

  Regarding \( \funcProd \): note that the cost of evaluating a term
  \( \funcProd\ F\ n\ m \) depends not only on the sizes of the arguments \( F \),
  \( n\) and \( m \), but also on the \emph{behaviour} of \( F \): if \( F \) is a
  function that computes its result in linear time in the size of the input, then
  \( \funcProd\ F\ n\ m \) will reach a normal form much faster than if \( F \)
  operates in exponential time.  The computation time is also affected by the size
  of the output that \( \funcProd \) produces, as this is fed into \( \mult \).
  This is why \( \funcinterpretcost{\funcProd} \) takes both a cost function and a
  size function as argument.
  We can orient both rules by choosing
  \( \funcinterpretsize{\funcProd} = \fatlambda F x y.y * {\max(F(x),1)}^x \)
  and
  \( \funcinterpretcost{\funcProd} =
  \fatlambda F^\cost F^\size x y.3 * x * y * {\max(F^\size(x),1)}^{x+1} + x * F^\cost(x) + 2 * x + 1 \).

  Notice that \( \funcinterpretcost{\funcProd} \) is not polynomial,
  but that is allowed in the general case.
\end{example}

\begin{example}\label{ex:costversuscostprime}
Consider an STRS with base types \( \sortset = \{ \nat \} \), function symbols
\(
  \zero \hasType \nat,
  \suc \hasType \nat \arrtype \nat,
  \minus \hasType \nat \arrtype \nat \arrtype \nat
\), and rules:
\[
\begin{array}{rclcrcl}
\minus\ (\suc\ x)\ (\suc\ y) & \to & \minus\ x\ y & \quad\quad &
\minus\ x\ \zero & \to & x \\
\end{array}
\]
If we choose \( \funcinterpretsize{\zero} \), \( \funcinterpretsize{\suc} \),
\( \funcinterpretcost{\zero} \) and \( \funcinterpretcost{\suc} \) like in
\cref{ex:addinterpret}, and \( \funcinterpretcost{\minus} = \fatlambda xy.y \),
then we have \( \totalcost{\minus\ \zero\ (\suc\ \zero)} = 1 \), while
\( \totalcostprime{\minus\ \zero\ (\suc\ \zero)} = 0 \).
This is because \( \minus\ \zero\ (\suc\ \zero) \) is in normal form, even though
it is not fully built from constructors.
\end{example}

\subsection{The Compatibility Theorem}

We now prove the Innermost Compatibility Theorem for STRSs.
The proof is analogous to that in~\cite{kop:vale:23} with some adaptations
to the definitions of \cref{subsec:interpretationdef} which
are possible here since our terms are applicative and second-order only; we also make some adaptations to the evaluation strategy
being innermost rather than call-by-value, in particular through the use of the \( \totalcostprime{\cdot} \) function.

\begin{restatable}[Innermost Compatibility]{theorem}{compat}\label{thm:compatibility}
  Suppose \( \trs \) is an STRS compatible with a \cs{} interpretation \( \aDom \)
  (following \cref{def:compatible}).
  Then for any valuations \( \alpha \) and \( \zeta \)
  we have
  \( \totalcostprime{s}_{\alpha,\zeta} > \totalcostprime{t}_{\alpha,\zeta} \) and
  \( \sizeinterpret{s}_\alpha \sizeGe \sizeinterpret{t}_\alpha \)
  whenever \( s \arrzR t \).
\end{restatable}

Hence, if we find a suitable \cs{} interpretation, which orients all the rules,
then we have a \emph{bound} on the derivation height of any given term: that is, the
number of reduction steps we need to take to reduce any given term to normal form.
For example, suppose $s_0 \arrzR s_1 \arrzR \dots \arrzR s_n$.
Let \( \alpha \) and \( \zeta \) map all variables either to \( 0 \) or to a constant
function \( \fatlambda \vec{x}.0 \).  Then \( \totalcostprime{s_0}_{\alpha,\zeta} \geq
\totalcostprime{s_1}_{\alpha,\zeta} + 1 \geq \dots \geq
\totalcostprime{s_n}_{\alpha,\zeta} + n \) --- and therefore, necessarily
\( \totalcostprime{s_0}_{\alpha,\zeta} \geq n \).
This holds whether \( s_0 \) has base type or not, since \( \totalcostprime{\cdot} \)
always maps to a natural number.
Observing that always \( \totalcost{s}_{\alpha,\zeta} \geq \totalcostprime{s}_{\alpha,
\zeta} \):

\begin{corollary}\label{cor:compatibility}
Suppose \( \trs \) is an STRS compatible with a \cs{} interpretation \( \aDom \),
and let \( s \) be a term with derivation height \( n \).
Then for any valuations \( \alpha \) and \( \zeta \),
\( \totalcost{s}_{\alpha,\zeta} \geq n \).
\end{corollary}

\medskip
In order to prove \cref{thm:compatibility},
we first stablish some useful lemmas.
Recall that in this paper,
all rules are of base type, i.e., they are fully applied.
Since reduction is innermost,
a rule may only be fired if the matching substitution
(i.e., the substitution \( \asub \) on the base case \( \ell \asub \arrz r \asub \)),
maps all variables to irreducible terms.
So in the lemmas below, without loss of generality,
we restrict ourselves to this type of substitutions and notice that
\( \totalcostprime{\avar \asub} = 0 \) for any variable \( \avar \).

\begin{lemma}\label{lem:casethingy}
For all terms \( s \app t \) with \( t \) of base type:
\( \sizeinterpret{s \app t}_\alpha = \sizeinterpret{s}_\alpha(\sizeinterpret{t}_\alpha) \) and
\( \costinterpret{s \app t}_{\alpha,\zeta} = \costinterpret{s}_{\alpha,\zeta}(\sizeinterpret{t}_{\alpha}) \)
for all \( \alpha,\zeta \).
\end{lemma}

\begin{proof}
The former statement holds by definition.
The latter holds by a straightforward case analysis on the form of \( s \):
it holds both if \( s = \avar \app s_1 \cdots s_n \)
and
\( s = \afun \app s_1 \cdots s_k \app t_1 \cdots t_n \).
For the \( \afun \) case, recall that we fixed
$\afun \hasType (\typeVec{\asort_1} \arrtype \bsort_1)
  \arrtype \cdots \arrtype (\typeVec{\asort_k} \arrtype \bsort_k) \arrtype
  \csort_1 \arrtype \cdots \arrtype \csort_l \arrtype \asort$
as the general type for function symbols \( \afun \), and since \( t \) has base
type, all higher-type arguments to \( \afun \) are necessarily already supplied.
Hence, \( \costinterpret{s}_{\alpha,\zeta} \) is indeed well-defined.
\end{proof}

Given a valuation \( \alpha \) and substitution \( \asub \),
we denote the \( \asub \)-extention of \( \alpha \)
by \( \alpha^\asub \) as the valuation defined by
\( \alpha^\asub(x) = \sizeinterpret{x\gamma}_\alpha \).
Let us start with some substitution lemmata.
\begin{lemma}\label{lemma:subst-size}
  Let \( \asub \) be a substitution mapping all variables to irreducible terms
  and \( \alpha \) be a valuation.
  Then, for any term \( \aterm \),
  \( \sizeinterpret{\aterm \asub}_\alpha = \sizeinterpret{\aterm}_{\alpha^\asub} \).
\end{lemma}
\begin{proof}
  By induction on the structure of \( \aterm \).
  \begin{itemize}
    \item If \( \aterm \) is a variable,
    we have \( \sizeinterpret{\avar \asub}_\alpha =
      \alpha^\asub(\avar) = \sizeinterpret{\avar}_{\alpha^\asub} \).
    \item If \( \aterm = \bterm \app \cterm \) is an application,
    we have
    \begin{align*}
      \sizeinterpret{(\bterm \app \cterm) \asub} &=
        \sizeinterpret{\bterm \asub}_{\alpha}(\sizeinterpret{\cterm \asub}_{\alpha})\\
      \overset{IH}&{=} \sizeinterpret{\bterm}_{\alpha^\asub}(\sizeinterpret{\cterm}_{\alpha^\asub}) %
      = \sizeinterpret{\bterm \app \cterm}_{\alpha^\asub}
      \qedhere
    \end{align*}
  \end{itemize}

\end{proof}

We can now prove the \emph{size} part of \cref{thm:compatibility}
which is stated as the lemma below.

\begin{lemma}\label{lemma:size-compatibility}
Let \( \trs \) be an STRS compatible with a \cs{} interpretation following \cref{def:compatible},
and let \( \aterm,\bterm \) be terms of the same type such that \( \aterm \arrzR \bterm \).
Then \(  \sizeinterpret{\aterm}_{\alpha} \sizeGe \sizeinterpret{\bterm}_{\alpha} \) holds
for all \( \alpha \).
\end{lemma}

\begin{proof}
By induction on the form of \( \aterm \).

In the base case, we have \( \aterm \arrzR \bterm \) by \( \ell \asub \arrz r \asub \).
Then we combine the substitution lemma (\cref{lemma:subst-size}) with the compatibility
requirement for size, i.e., \( \sizeinterpret{\ell}_\alpha \sizeGe \sizeinterpret{r}_\alpha \), as follows:
  \begin{align*}
      \sizeinterpret{\ell \asub}_\alpha = \sizeinterpret{\ell}_{\alpha^\asub}
      \sizeGe \sizeinterpret{r}_{\alpha^\asub}
      = \sizeinterpret{r \asub}_\alpha
  \end{align*}

In the application case, \( \aterm = \aterm_1 \app \aterm_2 \) and either
\( \bterm = \bterm_1 \app \aterm_2 \) with \( \aterm_1 \arrzR \bterm_1 \) or
\( \bterm = \aterm_1 \app \bterm_2 \) with \( \aterm_2 \arrzR \bterm_2 \).
In the first case, \( \sizeinterpret{\aterm_1}_\alpha \sizeGe \sizeinterpret{\bterm_1}_\alpha \)
by the induction hypothesis, which by definition gives
\( \sizeinterpret{\aterm}_\alpha = \sizeinterpret{\aterm_1}_\alpha(\sizeinterpret{\aterm_2}_\alpha)
\sizeGe \sizeinterpret{\bterm_1}_\alpha(\sizeinterpret{\aterm_2}_\alpha) =
\sizeinterpret{\bterm}_\alpha \).
In the second case, \( \sizeinterpret{\aterm_2}_\alpha \sizeGe \sizeinterpret{\bterm_2}_\alpha \) by
the induction hypothesis and
\( \sizeinterpret{\aterm}_\alpha = \sizeinterpret{\aterm_1}_\alpha(\sizeinterpret{\aterm_2}_\alpha)
\sizeGe \sizeinterpret{\aterm_1}_\alpha(\sizeinterpret{\bterm_2}_\alpha) =
\sizeinterpret{\bterm}_\alpha \) because by definition,
\( \sizeinterpret{\aterm_1}_\alpha \) is a \emph{weakly monotonic} function.
\end{proof}

Let us move on to cost versions of substitution lemmata.
First,
notice that we cannot directly define a \(\asub \)-extention for cost valuations.
Indeed, \( \costinterpret{\cdot}_{\alpha,\zeta} \) also depends on
a size valuation \( \alpha \).
So given a size valuation \( \alpha \),
we write \( \zeta^\asub_\alpha \) to denote the valuation
\( \zeta^\asub_\alpha \) with
\( \zeta^\asub_\alpha(x) = \costinterpret{\asub(x)}_{\alpha,\zeta} \).

\begin{lemma}\label{lem:cost-subst-base}
  Given \cs{} valuations \( \alpha, \gamma \) and a term \( \aterm \)
  of type
  order at most \( 1 \).
  Then
  \( \costinterpret{\aterm \asub}_{\alpha,\zeta} =
    \costinterpret{\aterm}_{\alpha^\asub,\zeta^\asub_\alpha} \).
\end{lemma}
\begin{proof} By induction on the size of \( \aterm \). We consider two cases:
  \begin{itemize}
    \item For the first case, we get \( \aterm = \avar \app \aterm_1 \ldots \aterm_n \),
    and
    \begin{itemize}
      \item If \( n = 0 \), we have
      \( \costinterpret{\avar \asub}_{\alpha,\zeta} = \zeta_\alpha^\asub(x) = \costinterpret{\avar}_{\alpha^\asub,\zeta_\alpha^\asub} \) by definition.
      \item If \( n > 0 \), we have
      \begin{align*}
        \costinterpret{(\avar \app \aterm_1 \ldots \aterm_n) \asub}_{\alpha,\zeta}
        &= \costinterpret{
          (\avar \asub) \app (\aterm_1 \asub) \ldots (\aterm_n \asub)
          }_{\alpha,\zeta}\\
        \overset{\text{\cref{lem:casethingy}}}&{=} \costinterpret{\avar \asub}_{\alpha,\zeta}(
            \sizeinterpret{\aterm_1 \asub}_{\alpha},
            \ldots,
            \sizeinterpret{\aterm_n \asub}_\alpha
        )\\
        \overset{\text{\cref{lemma:subst-size}}}&{=} \costinterpret{\avar \asub}_{\alpha,\zeta}(
            \sizeinterpret{\aterm_1}_{\alpha^\asub},
            \ldots,
            \sizeinterpret{\aterm_n}_{\alpha^\asub}
        )\\
        &= \zeta_\alpha^\asub(\avar)(
            \sizeinterpret{\aterm_1}_{\alpha^\asub},
            \ldots,
            \sizeinterpret{\aterm_n}_{\alpha^\asub}
        )\\
        &= \costinterpret{(\avar \app \aterm_1 \ldots \aterm_n)}_{\alpha^\asub,\zeta_\alpha^\asub}
      \end{align*}
    \end{itemize}
    \item For the second case we have \( \aterm = \afun \app \aterm_1 \ldots \aterm_k \app \bterm_1 \ldots \bterm_n \).
    Recall that we fixed
    $\afun \hasType (\typeVec{\asort_1} \arrtype \bsort_1)
      \arrtype \cdots \arrtype (\typeVec{\asort_k} \arrtype \bsort_k) \arrtype
      \csort_1 \arrtype \cdots \arrtype \csort_l \arrtype \asort$
    as the general type for \( \afun \).
    Hence, since we consider \( \aterm \) of type order at most 1,
    \( \afun \) must take at least \( k \) arguments, and \( 0 \leq n \leq l \).
    \begin{align*}
      &\costinterpret{(\afun \app \aterm_1 \ldots \aterm_k \app \bterm_1 \ldots \bterm_n) \asub}_{\alpha,\zeta}\\
      &=
      \costinterpret{\afun \app (\aterm_1 \asub) \ldots (\aterm_k \asub) \app (\bterm_1 \asub) \ldots (\bterm_n \asub)}_{\alpha,\zeta}\\
      &= \funcinterpretcost{\afun}(
        \costinterpret{\aterm_1\asub}_{\alpha,\zeta},
        \sizeinterpret{\aterm_1\asub}_{\alpha},
        \ldots,
        \costinterpret{\aterm_k\asub}_{\alpha,\zeta},
        \sizeinterpret{\aterm_k\asub}_{\alpha},
        \sizeinterpret{\bterm_1\asub}_\alpha,
        \ldots,
        \sizeinterpret{\bterm_n\asub}_\alpha
      )\\
      \overset{\text{IH, \cref{lemma:subst-size}}}&{=} \funcinterpretcost{\afun}(
        \costinterpret{\aterm_1}_{\alpha^\asub,\zeta_\alpha^\asub},
        \sizeinterpret{\aterm_1}_{\alpha^\asub},
        \ldots,
        \costinterpret{\aterm_k}_{\alpha^\asub,\zeta_\alpha^\asub},
        \sizeinterpret{\aterm_k}_{\alpha^\asub},
        \sizeinterpret{\bterm_1}_{\alpha^\asub},
        \ldots,
        \sizeinterpret{\bterm_n}_{\alpha^\asub}
      )\\
      &= \costinterpret{\afun \app \aterm_1 \ldots \aterm_k \app \bterm_1 \ldots \bterm_n}_{\alpha^\asub,\zeta_\alpha^\asub}
       \qedhere
    \end{align*}
  \end{itemize}

\end{proof}

Next, we connect the relationship between the two cost functions
we defined.

\begin{lemma}\label{lemma:cost-gtecost-prime}
  For any term \( \aterm \) of type order \( 0 \) or \( 1 \), and any
  substitution \( \asub \) such that all \( \asub(x) \) are in normal form, we have that
  \( \totalcost{\aterm}_{\alpha^\asub, \zeta_\alpha^\asub} \costGe \totalcostprime{\aterm \asub}_{\alpha,\zeta}\).
\end{lemma}
\begin{proof}
  We again consider two cases:
  \begin{itemize}
    \item For the first case, let \( \aterm = \avar \app \aterm_1 \ldots \aterm_n \).
      If \( n = 0 \) then \( \totalcost{\avar}_{\alpha^\asub, \zeta_\alpha^\asub} = 0 \)
      by definition, and since we assumed that \( \asub(\avar) \) is in normal form, also
      \( \totalcostprime{\avar\asub}_{\alpha,\zeta} = 0 \).
      If \( n > 0 \) and \( \aterm \) has base type, then
      \(  \totalcost{\aterm}_{\alpha^\asub, \zeta_\alpha^\asub} =
      \zeta_\alpha^\asub(\avar)(\sizeinterpret{\aterm_1}_{\alpha^\asub, \zeta_\alpha^\asub},\dots,\sizeinterpret{\aterm_n}_{\alpha^\asub, \zeta_\alpha^\asub}) +
      \Sigma_{i=1}^n \totalcost{\aterm_i}_{\alpha^\asub, \zeta_\alpha^\asub} =
      \costinterpret{\gamma(\avar)}_{\alpha,\zeta}(\sizeinterpret{\aterm_1\gamma}_{\alpha},\dots,\sizeinterpret{\aterm_n\gamma}_\alpha) +
      \Sigma_{i=1}^n \totalcost{\aterm_i}_{\alpha^\asub, \zeta_\alpha^\asub} \) by Lemmas~\ref{lemma:subst-size} and~\ref{lem:cost-subst-base}.
      By Lemma~\ref{lem:casethingy} and the induction hypothesis this
      \( \geq \costinterpret{(\avar\gamma) \app (\aterm_1\gamma) \cdots (\aterm_n\gamma)} +
      \Sigma_{i=1}^n \totalcostprime{\aterm_i\gamma}_{\alpha,\zeta} \).
      Since \( \avar\gamma \) is in normal form, either this is exactly \( \totalcostprime{\aterm\gamma} \), or
      \( \totalcostprime{\aterm\gamma} = 0 \) and we are done regardless.
      If \( n > 0 \) and \( \aterm \) does not have base type, we complete quickly with the induction hypothesis.
    \item For the second case, let \( \aterm = \afun \app \aterm_1 \ldots
      \aterm_k \app \bterm_1 \ldots \bterm_n \)
  We have two cases whether \( \aterm \asub \) is in normal form or not.
  In the first case, \( \totalcostprime{\aterm \asub}_{\alpha,\zeta} = 0 \)
  and certainly \( \totalcost{\aterm }_{\alpha^\asub,\zeta_\alpha^\asub} \geq 0 \).
  For the second case, \( \aterm \) is not in normal form.

  If \( \aterm \) has base type, then:
  \begin{align*}
    \totalcost{\aterm}_{\alpha^\asub,\zeta_\alpha^\asub}
    &= \totalcost{\afun \app \aterm_1 \ldots \aterm_k \app \bterm_1 \ldots \bterm_n}_{\alpha^\asub,\zeta_\alpha^\asub} \\
    &=\costinterpret{\aterm}_{\alpha^\asub,\zeta_\alpha^\asub} +
      \sum\limits_{i = 1}^k \totalcost{\aterm_i}_{\alpha^\asub,\zeta_\alpha^\asub}
    +
    \sum\limits_{j = 1}^n
      \totalcost{\bterm_j}_{\alpha^\asub,\zeta_\alpha^\asub}\\
    \overset{\text{\cref{lem:cost-subst-base}}}&{=}
      \costinterpret{\aterm \asub}_{\alpha,\zeta} +
      \sum\limits_{i = 1}^k \totalcost{\aterm_i}_{\alpha^\asub,\zeta_\alpha^\asub}
    +
    \sum\limits_{j = 1}^n
      \totalcost{\bterm_j}_{\alpha^\asub,\zeta_\alpha^\asub}\\
    \overset{IH}&{\geq} \costinterpret{\aterm \asub}_{\alpha,\zeta} +
    \sum\limits_{i = 1}^k \totalcostprime{\aterm_i \asub}_{\alpha,\zeta}
    +
    \sum\limits_{j = 1}^n
      \totalcostprime{\bterm_j \asub}_{\alpha,\zeta}\\
    &= \totalcostprime{\aterm \asub}_{\alpha,\zeta}
  \end{align*}
  If not, then:
  \begin{align*}
    \totalcost{\aterm}_{\alpha^\asub,\zeta_\alpha^\asub}
    &= \sum\limits_{i = 1}^k \totalcost{\aterm_i}_{\alpha^\asub,\zeta_\alpha^\asub}
    +
    \sum\limits_{j = 1}^n
      \totalcost{\bterm_j}_{\alpha^\asub,\zeta_\alpha^\asub}\\
    \overset{IH}&{\geq}
    \sum\limits_{i = 1}^k \totalcostprime{\aterm_i \asub}_{\alpha,\zeta}
    +
    \sum\limits_{j = 1}^n
      \totalcostprime{\bterm_j \asub}_{\alpha,\zeta}\\
    &= \totalcostprime{\aterm \asub}_{\alpha,\zeta}
    \qedhere
  \end{align*}
  \end{itemize}
\end{proof}

\begin{lemma}\label{lemma:cost-arr-ty-costge}
      Let \( \trs \) be an STRS compatible with a \cs{} interpretation following \cref{def:compatible}
      and \( \aterm,\bterm \) be type-1 terms of the same type, such that
      \( \aterm \arrzR \bterm \).
      Then we have that \( \costinterpret{\aterm}_{\alpha,\zeta} \costGe \costinterpret{\bterm}_{\alpha,\zeta} \).
\end{lemma}

\pagebreak[5]

\begin{proof}
By induction on the size of \( \aterm \).  There are two cases.

\begin{itemize}
\item First, \( \aterm = \avar  \app \aterm_1 \ldots \aterm_n \).
  We can write \( \avar \hasType \asort_1 \arrtype \cdots \arrtype \asort_k \arrtype \bsort \) and,
  since \( \aterm \) is type-1 (so \emph{not} of base type), \( n < k \).
  Then \( \bterm = \avar \app \aterm_1 \ldots \aterm_i' \ldots \aterm_n \)
  with \( \aterm_i \arrzR \aterm_i' \).
  Hence, by \cref{lemma:size-compatibility} and monotonicity of \( \zeta(\avar) \):
  \begin{align*}
    \costinterpret{\avar \app \aterm_1 \ldots \aterm_i \ldots \aterm_n}_{\alpha,\zeta}
    &= \zeta(\avar)(\sizeinterpret{\aterm_1}_\alpha,
    \ldots, \sizeinterpret{\aterm_i}_\alpha, \ldots, \sizeinterpret{\aterm_n}_\alpha)\\
    &\sizeGe
    \zeta(\avar)(\sizeinterpret{\aterm_1}_\alpha,
    \ldots, \sizeinterpret{\aterm_i'}_\alpha, \ldots, \sizeinterpret{\aterm_n}_\alpha)\\
    &= \costinterpret{\avar \app \aterm_1 \ldots \aterm_i' \ldots \aterm_n}_{\alpha,\zeta}
  \end{align*}

\item For the second part, \( \aterm = \afun \app \aterm_1 \cdots \aterm_k \app \bterm_1 \cdots \bterm_n \),
  recall that rules are of base type, and therefore reduction does not occur at head position.
  Hence, either \( \bterm = \afun \app \aterm_1 \cdots \aterm_i' \cdots \aterm_k \app \bterm_1 \cdots \bterm_n \)
  with \( \aterm_i \arrzR \aterm_i' \),
  or \( \bterm = \afun \app \aterm_1 \cdots \aterm_k \app \bterm_1 \cdots \bterm_i' \cdots \bterm_n \)
  with \( \bterm_i \arrzR \bterm_i' \).

  In the first case,
  \[
  \begin{array}{ll}
  & \costinterpret{\afun \app \aterm_1 \ldots \aterm_i \ldots \aterm_k \app \bterm_1 \ldots \bterm_n}_{\alpha,\zeta} \\
  = & \funcinterpretcost{\afun}(\costinterpret{\aterm_1}_{\alpha,\zeta},\sizeinterpret{\aterm_1}_\alpha,
    \ldots,\costinterpret{\aterm_i}_{\alpha,\zeta},\sizeinterpret{\aterm_i}_\alpha, \ldots,
    \costinterpret{\aterm_k}_{\alpha,\zeta},\sizeinterpret{\aterm_k}_\alpha,\sizeinterpret{\bterm_1}_\alpha,
    \dots,\sizeinterpret{\bterm_n}_\alpha)\\
  \costGe & \funcinterpretcost{\afun}(\costinterpret{\aterm_1}_{\alpha,\zeta},\sizeinterpret{\aterm_1}_\alpha,
    \ldots,\costinterpret{\aterm_i}_{\alpha,\zeta},\sizeinterpret{\aterm_i}_\alpha, \ldots,
    \costinterpret{\aterm_k}_{\alpha,\zeta},\sizeinterpret{\aterm_k}_\alpha,\sizeinterpret{\bterm_1}_\alpha,
    \dots,\sizeinterpret{\bterm_n}_\alpha) \\
  = & \costinterpret{\afun \app \aterm_1 \ldots \aterm_i' \ldots \aterm_k \app \bterm_1 \ldots \bterm_n}_{\alpha,\zeta} \\
  \end{array}
  \]
  because \( \sizeinterpret{\aterm_i}_\alpha \sizeGe \sizeinterpret{\aterm_i'} \) by
  \cref{lemma:size-compatibility},
  \( \costinterpret{\aterm_i}_{\alpha,\zeta} \costGe \costinterpret{\aterm_i'}_{\alpha,\zeta} \)
  by the induction hypothesis,
  and \( \funcinterpretcost{\afun} \) is weakly monotonic.
  For the second case, we only need \cref{lemma:size-compatibility} and monotonicity, not the induction hypothesis.
  \qedhere
\end{itemize}
\end{proof}

Finally, we are ready to prove the innermost compatibility theorem.
\compat*

\begin{proof}
  The second part of the theorem is given by \cref{lemma:size-compatibility}.
  The first part follows by induction on the reduction \( \aterm \arrzR \bterm \).

  \begin{itemize}
    \item For the base case, we have by \cref{lem:cost-subst-base,lemma:cost-gtecost-prime} that
    \[
      \totalcostprime{\ell \asub}_{\alpha,\zeta} = \costinterpret{\ell \asub}_\zeta =
      \costinterpret{\ell}_{\zeta_\alpha^\asub} \costGt \totalcost{r}_{\alpha^\asub,\zeta_\alpha^\asub} \costGe
      \totalcostprime{r \asub}_{\alpha,\zeta}
    \]
    \item For the application case with a variable root symbol, we have that
      \( \avar \app t_1 \ldots t_i \ldots t_n \arrzR \avar \app t_1 \ldots
      t_i' \ldots t_n \) with \( t_i \arrzR t_i' \).
      By induction we get \( \totalcostprime{t_i}_{\alpha,\zeta} \costGt \totalcostprime{t_i'}_{\alpha,\zeta} \),
      and by \cref{lemma:size-compatibility} we have \( \sizeinterpret{t_i}_\alpha \sizeGe \sizeinterpret{t_i'}_\alpha \).
      Then, if \( \aterm \) has base type:
      \begin{align*}
        &\totalcostprime{\avar \app t_1 \ldots t_i \ldots t_n}_{\alpha,\zeta}\\
        &= \costinterpret{\avar \app t_1 \ldots t_i \ldots t_n}_{\alpha,\zeta} +
        \sum\limits_{j = 1}^n \totalcostprime{t_j}_{\alpha,\zeta}\\
        &=\zeta(\avar)(\sizeinterpret{t_1}_\alpha, \ldots, \sizeinterpret{t_i}_\alpha, \ldots, \sizeinterpret{t_n}_\alpha) +
        (\sum\limits_{\begin{subarray}{l} j = 1 \dots n \\ j \neq i\end{subarray}}
        \totalcostprime{t_j}_{\alpha,\zeta}\ ) + \totalcostprime{t_i}_{\alpha,\zeta}\\
        &\costGe \zeta(\avar)(\sizeinterpret{\bterm_1}_\alpha, \ldots, \sizeinterpret{\bterm_i'}_\alpha, \ldots, \sizeinterpret{\bterm_n}_\alpha) +
        (\sum\limits_{\begin{subarray}{l} j = 1 \dots n \\ j \neq i\end{subarray}}
        \totalcostprime{t_j}_{\alpha,\zeta}\ ) +
        \totalcostprime{t_i}_{\alpha,\zeta},\\
        &\costGt \zeta(\avar)(\sizeinterpret{\bterm_1}_\alpha, \ldots, \sizeinterpret{\bterm_i'}_\alpha, \ldots, \sizeinterpret{\bterm_n}_\alpha) +
        (\sum\limits_{\begin{subarray}{l} j = 1 \dots n \\ j \neq i\end{subarray}}
        \totalcostprime{t_j}_{\alpha,\zeta}\ ) +
        \totalcostprime{t_i'}_{\alpha,\zeta}\\
        &= \totalcostprime{\avar \app t_1 \ldots t_i' \ldots t_n}_{\alpha,\zeta}
      \end{align*}
      If \( \aterm \) does not have base type, we have a similar reasoning without
      the component \( \costinterpret{\avar \app t_1 \ldots t_i \ldots t_n}_{\alpha,\zeta} \).
    \item For the application case with a function root symbol where the
      reduction is done in a base-type argument, we have that
      \( \afun \app s_1 \ldots s_k \app t_1 \ldots t_i \ldots t_n \arrzR
      \afun \app s_1 \ldots s_k \app t_1 \ldots t_i' \ldots t_n \) with
      \( t_i \arrzR t_i' \).
      Let us write \( \vec{s} \) for \( s_1 \ldots s_k \) and \( \cost(s) \)
      for \( \sum\limits_{j = 1}^k \totalcostprime{s_i}_{\alpha,\zeta} \)  below.
      We also abuse notation and write \( \costinterpret{\vec{s}}_{\alpha,\zeta},
      \sizeinterpret{\vec{s}}_\alpha \) for \( \costinterpret{s_1}_{\alpha,\zeta},
      \sizeinterpret{s_1}_\alpha,\dots,\costinterpret{s_k}_{\alpha,\zeta},
      \sizeinterpret{s_k}_\alpha \).
      If \( \aterm \) has base type:
      \begin{align*}
        &\totalcostprime{\afun \app \vec{s} \app t_1 \ldots t_i \ldots t_n}_{\alpha,\zeta} \\
        &= \costinterpret{\afun \app \vec{s} \app t_1 \ldots t_i \ldots t_n}_{\alpha,\zeta}
        + \cost(s) + \sum\limits_{j = 1}^n \totalcostprime{t_j}_{\alpha,\zeta}\\
        &= \funcinterpret{\afun}^\cost({{\costinterpret{\vec{s}}_{\alpha,\zeta}}}, \sizeinterpret{\vec{s}}_\alpha, \sizeinterpret{t_1}_\alpha, \ldots, \sizeinterpret{t_i}_{\alpha}, \ldots, \sizeinterpret{t_n}_\alpha)
        + \cost(s) + \sum\limits_{j = 1}^n \totalcostprime{t_j}_{\alpha,\zeta}\\
        &\costGe
        \funcinterpret{\afun}^\cost({{\costinterpret{\vec{s}}_{\alpha,\zeta}}}, \sizeinterpret{\vec{s}}_\alpha, \sizeinterpret{t_1}_\alpha, \ldots, \sizeinterpret{t_i'}_{\alpha}, \ldots, \sizeinterpret{t_n}_\alpha)
        + \cost(s) + \sum\limits_{j = 1}^n \totalcostprime{t_j}_{\alpha,\zeta}\\
        &\costGt \totalcostprime{\afun \app \vec{s} \app t_1 \ldots t_i' \ldots t_n}_{\alpha,\zeta}
      \end{align*}
      where in the last step we use \( \totalcostprime{t_i}_{\alpha,\zeta} \costGt
      \totalcostprime{t_i'}_{\alpha,\zeta} \),
      given by the IH\@.
      If \( s \) does not have base type, the reasoning is similar, only omitting the
      component \( \costinterpret{\afun \app \vec{s} \app t_1 \ldots t_i \ldots t_n}_{\alpha,\zeta} \).
    \item For the application case with a function root symbol where the
      reduction is done in a higher-type argument, we have that
      \( \afun \app s_1 \ldots s_i \ldots s_k \app t_1 \ldots t_n \arrzR
      \afun \app s_1 \ldots s_i' \ldots s_k \app t_1 \ldots t_n \) with
      \( s_i \arrzR s_i' \).
      Recall that by IH we get \( \totalcostprime{\aterm_i}_{\alpha,\zeta} \costGt
      \totalcostprime{\aterm_i'}_{\alpha,\zeta} \).
      If \( \aterm \) does not have base type, we conclude in a similar reasoning
      to the ones used above, simply counting the \( \totalcostprime{\cdot} \) for
      all arguments.
      If \( \aterm \) does have base type, note that \( s_i \) and \( s_i' \) are
      type-1 terms, so we can apply \cref{lemma:cost-arr-ty-costge} to obtain
      \( \costinterpret{\aterm_i} \costGe \costinterpret{\aterm_i'} \).
      With this in hand we reason as follows:
      \begin{align*}
        &\totalcostprime{\afun \app \aterm_1 \ldots \aterm_i \ldots \aterm_k \app \vec{t}}_{\alpha,\zeta}\\
        &= \funcinterpret{\afun}^\cost
        (\costinterpret{\aterm_1}_{\alpha,\zeta}, \sizeinterpret{\aterm_1}_\alpha,
          \ldots,
          \costinterpret{\aterm_i}_{\alpha,\zeta}, \sizeinterpret{\aterm_i}_\alpha
          \ldots,
          \costinterpret{\aterm_k}_{\alpha,\zeta}, \sizeinterpret{\aterm_k}_\alpha,
          \sizeinterpret{\vec{t}}_\alpha
        )\\
        &\phantom{AAAA} + (\sum\limits_{j = 1\ldots k,j \neq i} \totalcostprime{s_j}_{\alpha,\zeta}\ ) + \totalcostprime{s_i}_{\alpha,\zeta} +
          (\ \sum\limits_{j = 1}^n \totalcostprime{t_j}_{\alpha,\zeta}\ )\\
        &\text{By monotonicity of \( \funcinterpret{\afun}^\cost \) and \( \costinterpret{\aterm_i}_{\alpha,\zeta} \costGe \costinterpret{\aterm_i'}_{\alpha,\zeta} \),
        \( \sizeinterpret{\aterm_i}_{\alpha} \sizeGe \sizeinterpret{\aterm_i'}_{\alpha} \), we get}\\
        &\costGe
        \funcinterpret{\afun}^\cost
        (\costinterpret{\aterm_1}_{\alpha,\zeta}, \sizeinterpret{\aterm_1}_\alpha,
          \ldots,
          \costinterpret{\aterm_i'}_{\alpha,\zeta}, \sizeinterpret{\aterm_i'}_\alpha
          \ldots,
          \costinterpret{\aterm_k}_{\alpha,\zeta}, \sizeinterpret{\aterm_k}_\alpha,
          \sizeinterpret{\vec{t}}_\alpha
        )\\
        &\phantom{AAAA} +
        (\sum\limits_{j = 1\ldots k,j \neq i} \totalcostprime{s_j}_{\alpha,\zeta}\ )
        + \totalcostprime{s_i}_{\alpha,\zeta}
        + (\ \sum\limits_{j = 1}^n \totalcostprime{t_j}_{\alpha,\zeta}\ ) \\
        &\costGt
        \funcinterpret{\afun}^\cost
        (\costinterpret{\aterm_1}_{\alpha,\zeta}, \sizeinterpret{\aterm_1}_\alpha,
          \ldots,
          \costinterpret{\aterm_i'}_{\alpha,\zeta}, \sizeinterpret{\aterm_i'}_\alpha
          \ldots,
          \costinterpret{\aterm_k}_{\alpha,\zeta}, \sizeinterpret{\aterm_k}_\alpha,
          \sizeinterpret{\vec{t}}_\alpha
        )\\
        &\phantom{AAAA} +
        (\sum\limits_{j = 1\ldots k,j \neq i} \totalcostprime{s_j}_{\alpha,\zeta}\ )
        + \totalcostprime{s_i'}_{\alpha,\zeta}
        + (\ \sum\limits_{j = 1}^n \totalcostprime{t_j}_{\alpha,\zeta}\ ) \\
        &= \totalcostprime{\afun\app \aterm_1 \ldots \aterm_i' \ldots \aterm_k \app \vec{t}}_{\alpha,\zeta}
        \qedhere
      \end{align*}
  \end{itemize}

\end{proof}

\section{
  From Higher-Order Rewriting to
  \texorpdfstring{\( \bffT \)}{BFF} and Back
}%
\label{sec:rw-bff}

The main result of this paper roughly states that \( \bffT \)
consists exactly of those type-2 functionals computed by
an STRS compatible with a \cs{} tuple interpretation, whose main function is
polynomially bounded.
To formally state this result,
we must first define what it means for an STRS to \textit{compute} a type-2 functional and define precisely the class of \cs{} interpretations we are interested in.
The challenge is that such an STRS is required to do what it is supposed to do \emph{for every} input, and the inputs not only consist of words, but also of functions on words.
We thus have to find a way to encode type-1 functions seen as inputs to a type-2 program.

Indeed,
let us start by encoding words in \( \bffDom \) as terms.
We let \( \bit,\word \in \sortset \) and introduce symbols
\( \bito,\biti \hasType \bit \) and \( \nil \hasType \word,\ \consInfix \hasType
\bit \arrtype \word \arrtype \word \) (where \( \consInfix \) will be used in
infix notation).
Then for instance \( 001 \) is encoded as the term
\( \bito \consInfix (\bito \consInfix (\biti \consInfix \nil)) \).
In practice, we usually employ the cleaner list-like notation
\( [\bito ; \bito ; \biti ] \).
Let \( \encode{w} \) denote the
term encoding of a word \( w \).

\newcommand{\addAux}{\mathbin{\defFont{aux}}}
\newcommand{\xorF}{\mathbin{\defFont{xor}}}
\newcommand{\andF}{\mathbin{\defFont{and}}}
\newcommand{\orF}{\mathbin{\defFont{or}}}

\begin{example}[Implementing Binary Addition]\label{bff:implement-binary}
  Let us implement binary addition.
  For this purpose, we consider binary sequences written in \textit{little-endian}
  format, i.e., the least significant digit is at the head of the list.
  So the decimal number \( 6 \) is written as \( \mathtt{011} \) in little-endian notation.
  To start, we will need the following logical operations on bit symbols.
  \begin{align*}
    \bito \xorF \bito &\arrz \bito &
    \biti \xorF \biti &\arrz \bito &
    \bito \xorF \biti &\arrz \biti &
    \biti \xorF \bito &\arrz \biti \\
    \bito \andF \bito &\arrz \bito &
    \biti \andF \biti &\arrz \biti &
    \bito \andF \biti &\arrz \bito &
    \biti \andF \bito &\arrz \bito \\
    \bito \orF \bito  &\arrz \bito &
    \biti \orF \biti  &\arrz \biti &
    \bito \orF \biti  &\arrz \biti &
    \biti \orF \bito  &\arrz \biti
  \end{align*}
  The rules defining \( \addAux \hasType \word \arrtype \word \arrtype \bit \arrtype \word \)
  below compute the bitwise addition of two binary numbers, given a carrying value as the
  third argument.  This is done by recursing over the input lists.
  \begin{align*}
    &\addAux \app \nil \app \nil \app \bito \arrz \nil
    \\
    &\addAux \app \nil \app \nil \app \biti \arrz \biti \consInfix \nil
    \\
    &\addAux \app (a \consInfix as) \app \nil \app \accV \arrz (a \xorF \accV) \consInfix \addAux \app as \app \nil \app (a \andF \accV)
    \\
    &\addAux \app \nil \app (b \consInfix bs) \app \accV \arrz (b \xorF \accV) \consInfix \addAux \app \nil \app bs \app (b \andF \accV)
    \\
    &\addAux \app (a \consInfix as) \app (b \consInfix bs) \app \accV \arrz
      ((a \xorF b) \xorF \accV) \consInfix \addAux \app as \app bs \app (((a \xorF b) \andF \accV) \orF (a \andF b))
  \end{align*}
  Finally, we write the addition of binary numbers as the rule
  \( x \binAddInfix y \arrz \addAux \app x \app y \app \bito \).
  We can orient all rules by setting the following interpretation:
  \begin{align*}
  \funcinterpretsize{\bito} &= \funcinterpretsize{\biti} = \funcinterpretsize{\nil} = 0 &
  \funcinterpretcost{\bito} &= \funcinterpretcost{\biti} = \funcinterpretcost{\nil} = 0 \\
  \funcinterpretsize{\consInfix} &= \fatlambda x y.1 + y &
  \funcinterpretcost{\consInfix} &= \fatlambda x y.0 \\
  \funcinterpretsize{\mathit{op}} &= \fatlambda x y.0 &
  \funcinterpretcost{\mathit{op}} &= \fatlambda x y.1 & \text{for}\ \mathit{op} \in \{\xorF,\andF,\orF\} \\
  \funcinterpretsize{\addAux} &= \fatlambda x y a. 1 + {\max(x,y)} &
  \funcinterpretcost{\addAux} &= \fatlambda x y a. 1 + {7 * \max(x,y)} \\
  \funcinterpretsize{\binAddInfix} &= \fatlambda x y. 1 + {\max(x,y)} &
  \funcinterpretcost{\binAddInfix} &= \fatlambda x y. 2 + {7 * \max(x,y)}
  \end{align*}
  This interpretation assigns cost \( 0 \) to all ground constructor terms:
  ground terms that are exclusively built from the constructor symbols
  \( \biti \), \( \bito \), \( \nil \) and \( \consInfix \).
  The size of a binary word is its length.
  The cost of adding two numbers is bounded by \( 2 + 7 * \langle\text{the length of
  the longest}\rangle \): this accounts for iterating through the list, and doing 6
  bit operations for each step.
\end{example}

Next, we encode type-1 functions as a possibly infinite set of one-step
rewrite rules.

\begin{definition}\label{def:rw-oracle-simulation}
  Consider a type-1 function \( f : \bffDom \arrfunc \bffDom \) and
  let \( \funS_f \hasType \word \arrtype \word \) be a fresh function symbol.
  A set of rules \( \rules_f \)
  \defName{defines} a function
  \( f : \Nat \arrfunc \Nat \) by way of the symbol \( \funS_f \)
  if and only if
  \( \rules_f \) contains exactly the rules
  \( \funS_\afun \app \encode{n} \arrz \encode{m} \)
  where \( n,m \in \Nat \) and \( m = f(n) \).
\end{definition}

Intuitively,
this infinite set of rules is the rewriting counterpart of an oracle \( f \).
Indeed, in a single rewrite step \( \funS_\afun \app \encode{x} \)
rewrites to the value \( \encode{f(x)} \).
Henceforth,
we assume given that our STRS \( \trs \) at hand
is such that \( \syntaxSig \) contains
\( \bito, \biti, \nil, \consInfix \) typed as above
and a
distinguished
symbol \( \funF \hasType {(\word \arrtype\word)}^k \arrtype \word^l \arrtype \word \).
Given type-1 functions \( f_1, \dots, f_k \),
we write \( \syntaxSig_{\vec{f}} \) for \( \syntaxSig \) extended with function
symbols \( \funS_{f_i} \hasType \word \arrtype \word \), with \( 1 \leq i \leq k \),
and let \( \rulesExtdvec = \rules \cup \bigcup_{i=1}^k \rules_f \).
Now we can define the notion of type-2 computability for such STRSs.

\begin{definition}\label{def:type-2-rw-computability}
Let \( \trs \) be an STRS\@.
We say that \( \funF \) \defName{computes} the type-2 functional
\( \aTTFunc
\)
in \( (\syntaxSig,\rules) \) iff
for all type-1 functions
\( f_1, \dots, f_k
\) and all \( w_1,\dots,w_l \in \bffDom \),
\( \funF \app \funS_{f_1} \cdots \funS_{f_k} \app \encode{w_1} \cdots
\encode{w_l} \arrz_{\rulesExtdvec}^+ \encode{u} \), where \( u =
\aTTFunc(f_1,\dots,f_k,w_1,\dots,w_l) \).
\end{definition}

Next,
we define what we mean by polynomially bounded interpretation.

\begin{definition}\label{def:polynomially-bounded-int}
  We say that an STRS \( \trs \) \defName{admits} a polynomially bounded interpretation
  iff \( \trs \) is compatible with a \cs{} interpretation such that:
  \begin{itemize}
    \item \( \sizeInt{\word} = (\Nat,\geq) \);
    \item
      \( \funcinterpretcost{\bito} = \funcinterpretcost{\biti} =
      \funcinterpretcost{\nil} = 0 \),
      \( \funcinterpretcost{\consInfix} = \fatlambda xy.0\),
      and \( \funcinterpretsize{\consInfix} = \fatlambda xy.x + y +c \)
      for some \( c \geq 1 \);
    \item \( \funcinterpretcost{\funF} \) is bounded by a polynomial
    in \( \polySet{F_1^\cost,F_1^\size,\ldots,F_k^\cost,F_k^\size;x_1, \ldots,x_l} \).
    \end{itemize}
\end{definition}
Note that the condition on \( \funcinterpretsize{\consInfix} \) corresponds to what is called \textit{additive} quasi-interpretation in~\cite{DBLP:journals/tcs/BonfanteMM11}.
Finally, we can formally state our main result.

\begin{theorem}\label{thm:main}
A type-2 functional
\( \aTTFunc \)
is in \defName{\( \bffT \)}
if and only if
there exists a finite orthogonal STRS
\( \trs \)
such that the distinguished symbol \( \funF \) computes \( \aTTFunc \) in \( \trs \)
and \( \rules \) admits a polynomially bounded \cs{} interpretation.
\end{theorem}

Note that the polynomial bound \emph{only} refers to the cost-size interpretation for \( \funF \).
Indeed,
except for \( \bito \), \( \biti \), \( \nil \), and \( \consInfix \)
there are no requirements on the interpretations for other symbols.
We will prove this result in two parts.
First, we prove soundness in Section~\ref{sec:soundness}
which states that every type-2 functional computed by an STRS as above
is in \( \bffT \).
Then in Section~\ref{sec:completeness} we prove completeness,
which states that every functional in \( \bffT \) can be computed by such an STRS\@.
In order to simplify proofs,
we only consider type-2 functions of rank (1,1).
We claim that the results can be easily generalized, but the proofs become more
tedious when handling multiple arguments.

\begin{example}\label{example:nontrivial-bff}
  Let us consider the type-2 functional defined by
  \( \aTTFunc \coloneq \fatlambda f x. \sum_{i < |x|} f(i) \),
  where \( |x| \) refers to the length of the word \( x \).
  Notice that \( \aTTFunc \) adds all \( f(i) \) over each word \( i \in \bffDom \)
  whose value (as a binary number) is smaller than the length of \( x \).
  This functional was proved to lie in \( \bffT \) in~\cite{DBLP:journals/jfp/IrwinRK01},
  where the authors used an encoding of \( \aTTFunc \) as a
  \( \mathsf{BTLP}_2 \) program.
  We can encode \( \aTTFunc \) as an STRS as follows.
  We expand on the STRS from \cref{bff:implement-binary}.
  Let us also include the type \( \nat \) and data constructors \( \zero \)
  and \( \suc \) from \cref{ex:add}, and consider ancillary symbols
  \( \lengthOf \hasType \word \arrtype \nat \)
  and \( \tobin \hasType \nat \arrtype \word \), defined by the following rules:
  \begin{align*}
  \lengthOf \app \nil &\arrz \zero &
  \lengthOf \app (a \consInfix as) &\arrz \suc \app (\lengthOf \app as) \\
  \tobin \app \zero &\arrz \nil &
  \tobin \app (\suc\ n) &\arrz (\tobin \app n) \binAddInfix (\biti \consInfix \nil)
  \end{align*}
  The former computes the length of a given word and the latter
  converts a number from unary to binary representation (using the binary addition
  symbol \( \binAddInfix \hasType \word \arrtype \word \arrtype \word \), whose rules
  were given in \cref{bff:implement-binary}).
  Then \( \aTTFunc \) is computed by:
  \begin{align*}
    \compute \app \aFuncVar \app \zero \app \accV
    &\arrz
    \accV
    \\
    \compute \app \aFuncVar \app (\suc \app \counterV) \app \accV
    &\arrz
    \compute \app \aFuncVar \app \counterV \app
    (\accV \app \binAddInfix \app {F(\tobin \app \counterV)} )\\
    \main \app \aFuncVar \app \avar
    &\arrz
    \compute \app \aFuncVar \app (\lengthOf \app \avar) \app \nil
  \end{align*}
  That is,
  if we want to compute \( \aTTFunc(f,x) \) we simply reduce the term
  \( \main \app \funS_f \app \encode{x} \) to normal form.
  By \cref{thm:main}, to show that this system is in \( \bffT \) via our rewriting formalism
  we need to exhibit a \cs{} tuple interpretation for it that
  satisfies Definition~\ref{def:polynomially-bounded-int}.
  The interpretation functions for \( \bito \), \( \biti \), \( \nil \) and
  \( \consInfix \) from \cref{bff:implement-binary} satisfy the requirements
  (with $c = 1$); in addition let us set the following interpretations:
  \begin{align*}
  \funcinterpretsize{\zero} &= 0 &
  \funcinterpretcost{\zero} &= 0 &
  \funcinterpretsize{\suc} &= \fatlambda x.1+x &
  \funcinterpretcost{\suc} &= \fatlambda x.0 \\
  \funcinterpretsize{\lengthOf} &= \fatlambda x.x &
  \funcinterpretcost{\lengthOf} &= \fatlambda x.1+x &
  \funcinterpretsize{\tobin} &= \fatlambda x.1+x &
  \funcinterpretcost{\tobin} &= \fatlambda x.7 * x^2 + 3 * x + 1
  \end{align*}
  \begin{align*}
  \funcinterpretsize{\compute} &= \fatlambda F j a. j + \max(a,F(j)) \\
  \funcinterpretcost{\compute} &= \fatlambda F^\cost F^\size j a.
    4 * j^3 + 7 * j * \max(a, F^\size(j)) + j * F^\cost(j) + 1 \\
  \funcinterpretsize{\main} &= \fatlambda F x.x + F(x) \\
  \funcinterpretcost{\main} &= \fatlambda F^\cost F^\size x.
    4 * x^3 + 7 * x * F^\size(x) + x * F^\cost(x) + x + 3
  \end{align*}
  This orients all rules, and \( \funcinterpretcost{\main} \) is a polynomial.
\end{example}

\section{Soundness}\label{sec:soundness}

In order to prove soundness,
let us consider a fixed finite orthogonal STRS \( \rules \) with a distinguished
function symbol \( \funF \) which admits a polynomially bounded \cs{}
interpretation, %
that computes a type-2 functional \( \aTTFunc \).
We proceed to show that \( \aTTFunc \) is in \( \bffT \) roughly as follows.
\begin{enumerate}
  \item
  Since \( \rules \) computes \( \aTTFunc \) and admits a polynomially bounded
  interpretation,
  we show that so does the extended system \( \rulesExtd \)
  (Definition~\ref{def:type-2-rw-computability}).
  The restriction on \( \funcinterpretsize{\consInfix} \)
  (Definition~\ref{def:polynomially-bounded-int}) implies that
  \( \costinterpret{\funF \app \funS_f \app \encode{w}} \)
  is bounded by a second-order polynomial over \( |f|, |w| \).
  We show this in Lemma~\ref{lemma:rules-extd-compatible}.
  By compatibility (Corollary~\ref{cor:compatibility}),
  we can do at most polynomially many steps when reducing \( \funF \app \funS_f \app \encode{w} \).

  \item The cost polynomial restricts the size of any input that the
    function variable \( F \) is applied to (e.g., a cost bound
    of \( 3 + F^\cost(m) \) implies that \( F \) is never called on a term with
    size interpretation \( > m \)).
    This is the subject of Lemma~\ref{lemma:oracle-subterm-lemma}.
  \item Using the observations above,
  we then show that by graph rewriting we can simulate \( \rulesExtd \)
  and compute each \( \rulesExtd \)-reduction step in polynomial time on an OTM\@.
  This guarantees that \( \aTTFunc \) is in \( \bffT \), Theorem~\ref{thm:soundness}.
\end{enumerate}

\subsection{Interpreting The Extended STRS, Polynomially}\label{subsec:interpretextd}

Our first goal is to provide a polynomially bounded \cs{} interpretation
to the extended system \( \rulesExtd \).
We start with the observation that
the size interpretation of words in \( W \) is proportional to their length.
Indeed, since \( \funcinterpretsize{\consInfix} = \fatlambda xy.x+y+c \)
(Definition~\ref{def:polynomially-bounded-int})
let
\( \mu := \max(\funcinterpretsize{\bito},\funcinterpretsize{\biti}) + c \)
and
\( \nu := \funcinterpretsize{\nil} \).
Consequently, for all \( w \in \bffDom \):
\begin{equation}\label{encodebounds}
|w| \leq \sizeinterpret{\encode{w}} \leq \mu * |w| + \nu
\end{equation}

Recall that by Definition~\ref{def:rw-oracle-simulation} the extended system
\( \rulesExtd \) has possibly infinitely many rules of the form
\( \funS_f \app \encode{w} \arrz \encode{f(w)} \).
Such rules \( \funS_f \) represent calls for an oracle to compute \( f \) in a single step.
Thus, we set their cost to 1.
The size should be given by the length of the oracle output, taking the overhead
of interpretation into account.
Hence, we obtain:
\[
\funcinterpretcost{\sfSymbol} = \fatlambda x. 1
\quad\quad
\funcinterpretsize{\sfSymbol} = \fatlambda x. \mu * |f|(x) + \nu
\]
Recall that \( |f|(n) = \max\limits_{{|y|} \leq n} |f(y)| \).
Hence, \( |f| \) is weakly monotonic, and therefore so is \( \funcinterpretsize{\sfSymbol} \).
This orients the rules in \( \rules_f \) because
\( \costinterpret{\sfSymbol \app \encode{w}} = 1 > 0 = \totalcost{\encode{f(w)}}
\), and
\( \sizeinterpret{\sfSymbol \app \encode{w}} = \mu * |f|(\sizeinterpret{\encode{
  w}}) + \nu \geq \mu * |f|(|w|) + \nu \geq \mu * |f(w)| + \nu \) by definition of
  \( |f| \), which is superior or equal to \( \sizeinterpret{\encode{f(w)}} \).

\begin{lemma}\label{lemma:rules-extd-compatible}
There exists a second-order polynomial \( D \) so that
\( D(|f|,|w|) \) bounds the derivation height of \( \funF \app \sfSymbol \app
\encode{w} \) for any \(f \in \bffDom \arrfunc \bffDom \) and \( w \in
\bffDom \).
\end{lemma}

\begin{proof}
As \( \funcinterpretcost{\funF} \) is bounded by a second-order polynomial
\( \fatlambda F^c F^s x.P \), we can let \( D(F,n) :=
P(\fatlambda x.1,\fatlambda x.\mu*F(x)+\nu,\mu*n+\nu) \).  Then \( D \) is a
second-order polynomial, and \( D(|f|,|w|) \geq \funcinterpretcost{\funF}(
\funcinterpretcost{\sfSymbol},\funcinterpretsize{\sfSymbol},\sizeinterpret{
\encode{w}}) = \totalcost{\funF \app \sfSymbol \app \encode{w}} \)
(we omit \( \alpha \) and \( \zeta \) here because there are no variables in
\( \funF \app \sfSymbol \app \encode{w} \), and therefore the valuations are
not relevant to the \( \mathsf{cost} \) value).
By \cref{cor:compatibility}, this serves as a bound on the derivation height
of \( \funF \app \sfSymbol \app \encode{w} \).
\end{proof}

Notice that this lemma does not imply that \( \aTTFunc \) is in \( \bffT \).
It only guarantees that there is a polynomial bound to the
\textit{number of rewriting steps} for such systems.
However, it does not immediately follow that this number is a reasonable bound
for the actual computational cost of simulating a reduction on an OTM.
Consider for example a rule \( \afun \app (\suc \app n) \app t \arrz \afun \app n \app (\consFont{c} \app t \app t) \).
Every step doubles the size of the term.
A naive implementation --- which copies the duplicated term in each step --- would take exponential time.
Moreover, a single step using the oracle can create a very large output, which is not considered part of the cost of the reduction, even though an OTM would be unable to use it without first fully reading it.
Therefore,
in order to prove soundness,
we show how to realize a
reasonable implementation of rewriting w.r.t.~OTMs.
In essence, we will show that
(1) oracle calls are not problematic in the presence of polynomially bounded interpretations, and
(2) we can handle duplication with an appropriate representation of rewriting. This is very
much in the style of what has been done for first-order rewriting and the $\lambda$-calculus
in the past~\cite{AvanziniMoser,DalLagoMartini,AccattoliDalLago}.

\subsection{Bounding The Oracle Input}\label{subsec:oraclebound}

We first show that calling the oracle along a computation does not introduce
an exponential overhead along the way.
More precisely,
we will show that there
exists a second-order polynomial \( B \) such that if an oracle call
\( \sfSymbol \app \encode{x} \) occurs anywhere along the reduction \( \funF \app
\sfSymbol \app \encode{w} \arrzT \encode{v} \), then \( |x| \leq B(|f|,|w|) \).
From this, we know that the growth of the  overall term size during
an oracle call is at most \( |f|(B(|f|, |w|)) \).

Let \( P \) again be any polynomial bounding \( \funcinterpretcost{\funF} \).
Since \( P \) is a second-order polynomial, each occurrence of a sub-expression
\( F^c(E) \) in \( P \) is itself a second-order polynomial, and so is \( E \).
Let us enumerate these arguments as \( E_1, \dots, E_n \).
We can then construct the new polynomial \( Q \) as follows:
\[
Q \coloneq \sum\limits_i E_i \quad \text{where occurrences of}\ F^c(E_j)\
\text{inside}\ E_i\ \text{are replaced by}\ 1
\]

The idea here is that the polynomial \( Q \) sums up all those expressions \( E_i \)
(with \(1 \leq i \leq n \)) given to \( F^\cost \) as argument.  We do this because
intuitively, if during the computation a word \( \encode{v} \) is ever given to the oracle,
then the cost for the oracle computation must be accounted for in \( \funcinterpretcost{\funF} \).
Hence, \( \sizeinterpret{v} \) must be bounded by some \( E_i \), so certainly by the sum of all
\( E_i \).  We can safely replace occurrences of \( F^\cost(E_j) \) that occur inside another
\( E_i \) by \( 1 \) because \( F^\cost \) will be instantiated by \( \funcinterpretcost{\sfSymbol}
= \fatlambda x.1 \) (or, in Lemma~\ref{lem:technicalG}, by a function that maps
specific input to \( 1 \)).

Due to this replacement, \( F^\cost \) does not occur in \( Q \), so \( Q \) is only
parametrized by \( F^\size \) and some type-0 variable \( x \).
It is also possible for \( Q \) to be a constant polynomial, or to be parametrized
only by a type-0 variable.
Its final shape depends on the arguments provided
to the \( F^\cost \) in \( P \).
We let \( B(G,y) := Q(\fatlambda z.\mu * G(z) + \nu,\mu * y + \nu) \).

\begin{example}\label{example:building-poly-bounds-size}
  Let us illustrate the construction of such polynomial \( Q \).
  Consider the following polynomial \( P \).
  \[ P = \fatlambda F^c F^s x.x * F^\cost(3+F^\size(9*x)) + F^\cost(12) *
  F^\cost(3+x*F^\cost(2))+5, \]
  then \( Q \) is built by adding each of those arguments given to \( F^\cost \).
  We get
  \begin{align*}
    Q &= 3 + F^\size(9*x) + 12 + 3+x*1 + 2 \\
      &= 20 + F^\size(9*x) + x
  \end{align*}
  Finally, we construct the polynomial \( B \).
  \begin{align*}
    B(G,x) &= 20 + \mu * G(9 * (\mu * x + \nu)) + \nu + (\mu * x + \nu) \\
           &= 20 + 2 * \nu + G(9*\mu*x + 9*\nu) + \mu * x
  \end{align*}
\end{example}

Now \( B \) gives an upper bound to the argument values for \( F^\cost \) that
are considered: if a function differs from \( \funcinterpretcost{\sfSymbol} \)
only on argument values greater than \(  B(|f|,|w|) \), then we can use it in \( P \) and
obtain the same result.  Formally:

\begin{lemma}\label{lem:technicalG}
Fix \( f,w \).
Let \( G : \Nat \arrfunc \Nat \) with \( G(z) = \funcinterpretcost{\sfSymbol}(z)  \) if
\( z \leq B(|f|,|w|) \).
Then \( P(G,\funcinterpretsize{\sfSymbol},\sizeinterpret{\encode{w}}) =
P(\funcinterpretcost{\sfSymbol},\funcinterpretsize{\sfSymbol},\sizeinterpret{
\encode{w}}) \).
\end{lemma}

This is proved by induction on the form of \( P \), using that \( G \) is never
applied on arguments larger than \( B(|f|,|w|) \).
Lemma~\ref{lem:technicalG} is used in the following key result:

\begin{lemma}[Oracle Subterm Lemma]\label{lemma:oracle-subterm-lemma}
    Let \( f : \bffDom \arrfunc \bffDom \) be a type-1 function and
    \( w \in \bffDom \).
    If \( \funF \app \funS_f \app \encode{w}
    \arrz_{\rulesExtd}^* C[\sfSymbol \app \encode{x}] \) for some context
    \( C \), then \( |x| \leq B(|f|,|w|) \).
\end{lemma}

\begin{proof}
By way of a contradiction, suppose there exist \( f, w \), and \( x \) such that
\( \funF \app \sfSymbol \app \encode{w} \arrz_{\rulesExtd}^* C[\sfSymbol \app
\encode{x}] \) for some context \( C \), and \( |x| > B(|f|, |w|) \).
Let us now construct an alternative oracle:
let \( \zero \hasType \nat, \suc \hasType \nat \arrtype \nat,
\sfSymbol' \hasType \word \arrtype \word \) and \( \helper \hasType \nat \arrtype
\nat \arrtype \nat \), and for \( N := D(|f|,|w|) \), let \( \rules_{f,w}' \) be
given by:
\[
\begin{array}{rcllcrcl}
\sfSymbol' \app \encode{x} & \arrz & \encode{f(x)} & \text{if}\ |x| \leq
  B(|f|, |w|) & \quad\quad &
\helper \app \zero \app \bvar & \arrz & \bvar \\
\sfSymbol' \app \encode{x} & \arrz & \helper \app \unary{N} \app
  \encode{f(x)} & \text{otherwise} & \quad\quad &
\helper \app (\suc \app \avar) \app \bvar & \arrz &
  \helper \app \avar \app \bvar \\
\end{array}
\]
Where \( \unary{N} \) is the unary number encoding of \( N \),
as introduced in Section~\ref{subsec:hotr}.
Notice that by definition the rules for \( \sfSymbol' \) will produce
\( \encode{f(x)} \) in a single step if \( |x| \leq B(|f|, |w|) \)
but they will take \(N+2\) steps otherwise.
Also observe that \( \funS_f \) and \( \funS_f' \) behave the same; that is,
\( \funS_f \app \encode{x} \) and \( \funS_f' \app \encode{x} \) have the
same normal form on any input \( \encode{x} \).
We extend the interpretation function of the original signature with:
\[
\funcinterpretcost{\funS_f'} = \fatlambda \avar.\left\{
\begin{array}{ll}
1 & \text{if}\ \avar \leq B(|f|, |n|) \\
N+2 & \text{if}\ \avar > B(|f|, |n|) \\
\end{array}
\right.
\quad\quad
\funcinterpretsize{\funS_f'} = \funcinterpretsize{\funS_f}
\]
\[
\begin{array}{rclcrclcrclcrcl}
\funcinterpret{\helper}^\cost & = & \fatlambda x y. x + 1 & \quad &
\funcinterpret{\helper}^\size & = & \fatlambda x y. y & \quad &
\funcinterpret{\zero}^\size & = & 0 & \quad &
\funcinterpret{\suc}^\size & = & \fatlambda x. x + 1 \\
\end{array}
\]
We easily see that this orients all rules in \( \rules_{f,w} \).
Then, by Lemma~\ref{lem:technicalG}, \( \totalcost{\funF \app \sfSymbol' \app
\encode{w}} \leq P(\funcinterpretcost{\sfSymbol'},\funcinterpretsize{\sfSymbol'},
\sizeinterpret{\encode{w}}) = P(\funcinterpretcost{\sfSymbol},\funcinterpretsize{
\sfSymbol},\sizeinterpret{\encode{w}}) \leq D(|f|,|w|) = N \).
Yet, as we have \( \funF \app \sfSymbol \app \encode{w} \arrz_{\rulesExtd}^*
C[\sfSymbol \app \encode{x}] \), we also have \( \funF \app \sfSymbol \app
\encode{w} \arrz_{\rules \cup \rules_{f,w}'} C'[\sfSymbol' \app \encode{x}] \),
where \( C' \) is obtained from \( C \) by replacing all occurrences of
\( \sfSymbol \) by \( \sfSymbol' \).
Since \( |x| > B(|f|,|w|) \) by assumption, the reduction
\( \funF \app \sfSymbol' \app \encode{w} \arrz_{\rules \cup \rules_{f,w}'}^*
C[\sfSymbol' \app \encode{w}] \arrz_{\rules \cup \rules_{f,w'}}^*
C[\encode{f(x)}] \) takes strictly more than \( N \) steps,
contradicting Corollary~\ref{cor:compatibility}.
\end{proof}

\subsection{Graph Rewriting}\label{subsec:graphs}

Lemma~\ref{lemma:rules-extd-compatible} guarantees that if \( \rules \) is
compatible with a suitable interpretation, then at most polynomially many
\( \rulesExtd \)-steps can be performed starting in \( \funF \app \sfSymbol \app
\encode{w} \).
However, as observed in Section~\ref{subsec:interpretextd}, this does not yet
imply that a type-2 functional computed by an STRS with such an
interpretation is in \( \bff \).  To simulate a reduction on an OTM,
we must find a representation whose size does not increase too much in any
given step.
The answer is \emph{graph rewriting}.

\begin{definition}
A \defName{term graph} for a signature \( \signature \) is a tuple
\( (\gVert,\gLab,\gSucc,\gRoot) \) with
  \( \gVert \) a finite nonempty set of vertices;
  \( \gRoot \in \gVert \) a designated vertex called the \emph{root};
  \( \gLab : \gVert \arrfunc \signature \cup \{ \gApp \} \)
  a partial function with \( \gApp \) fresh; and
  \( \gSucc : \gVert \arrfunc \gVert^* \) a total function such that
  \( \gSucc(\aVert) = \aVert_1\aVert_2 \) when \( \gLab(\aVert) = \gApp \)
  and \( \gSucc(\aVert) = \gEmptyW \) otherwise.
We view this as a directed graph, with an edge from \( v \) to \( v' \)
if \( v' \in \gSucc(v) \), and require that this graph is \emph{acyclic} (i.e.,
there is no path from any \( v \) to itself).
Given term graph \( \aGraph \), we will often directly refer to
\( \gVert_\aGraph, \gLab_\aGraph \), etc.
\end{definition}

\begin{figure}[H]
\centering
\begin{subfigure}[b]{0.29\textwidth}
  \begin{tikzpicture}[node distance=0.3cm]
      \node[draw, circle] (0) at (0,0) {\(v_0 : \gApp \)};
      \node (1) [below left=of 0] {\( v_1 : \gApp \)};
      \node (2) [below =of 1] {\( v_2 : \add \)};
      \node (3) [right=of 2] {\( v_3 : \gVar \)};
      \node (4) [below right=of 0] {\( v_4 : \gVar \)};
      \draw [->] (0) -- (1);
      \draw [->] (1) -- (2);
      \draw [->] (1) -- (3);
      \draw [->] (0) -- (4);
  \end{tikzpicture}
\caption{}\label{fig:termgraph:full}
\end{subfigure}
\begin{subfigure}[b]{0.16\textwidth}
  \begin{tikzpicture}[node distance=0.3cm]
      \node [draw, circle] (5) at (0,0) {\( \gApp \)};
      \node (6) [below left=of 5] {\(\gApp \)};
      \node (7) [below=of 6] {\(\add \)};
      \node (8) [right=of 7] {\(\gVar \)};
      \node (9) [below right=of 5] {\(\gVar \)};
      \draw [->] (5) -- (6);
      \draw [->] (6) -- (7);
      \draw [->] (6) -- (8);
      \draw [->] (5) -- (9);
  \end{tikzpicture}
\caption{}\label{fig:termgraph:simplified}
\end{subfigure}
\begin{subfigure}[b]{0.21\textwidth}
  \begin{tikzpicture}[node distance=0.3cm]
      \node[root] (0) at (0,0) {\(\gApp \)};
      \node (1) [below left=of 0] {\( \gApp \)};
      \node (2) [below left=of 1] {\( \add \)};
      \node (3) [right=of 2] {\( \gApp \)};
      \node (tmp) [below right=of 0] {};
      \node (4) [below of=0, node distance=2cm] {\( \gVar \)};
      \node (5) [below left =of 3] {\( \suc \)};
      \draw[->] (0) -- (1);
      \draw[->] (1) -- (2);
      \draw[->] (1) -- (3);
      \draw[->] (0) -- (4);
      \draw[->] (3) -- (5);
      \draw[->] (3) -- (4);
  \end{tikzpicture}
  \caption{}\label{fig:termgraph:shared1}
\end{subfigure}
\begin{subfigure}[b]{0.2\textwidth}
  \begin{tikzpicture}[node distance=0.3cm]
      \node[root] (0) at (0,0)        {\(\gApp \)};
      \node (1) [below left=of 0]     {\( \gApp \)};
      \node (11) [below left=of 1]    {\( \afun \)};
      \node (12) [below right=of 1]   {\( \gApp \)};
      \node (121) [below left=of 12]  {\( \bfun \)};
      \node (122) [below right=of 12] {\( \gVar \)};
      \draw[->] (0) -- (1);
      \draw[->] (0) -- (12);
      \draw[->] (1) -- (11);
      \draw[->] (1) -- (12);
      \draw[->] (12) -- (121);
      \draw[->] (12) -- (122);
  \end{tikzpicture}
  \caption{}\label{fig:termgraph:shared2}
\end{subfigure}
\caption{A term graph, its simplified version, and two graphs with sharing}\label{fig:termgraph}
\end{figure}

Term graphs can be denoted visually in an intuitive way.  For example, using
\( \signature \) from Example~\ref{ex:add}, the graph with
\( \gVert = \{v_0,\dots,v_4\} \),
\( \gLab = \{ v_0,v_1 \mapsto \gApp,\ v_2 \mapsto \add \} \),
\( \gSucc = \{ v_0 \mapsto v_1v_4,\ v_1 \mapsto v_2v_3,\ v_3,v_4,v_5 \mapsto
  \gEmptyW \} \) amd
\( \gRoot = v_0 \) is pictured in \cref{fig:termgraph:full}.
We use \( \gVar \) to indicate unlabeled vertices and a circle for \( \gRoot \).
We will typically omit vertex names, as done in \cref{fig:termgraph:simplified}.
Note that the definition allows multiple vertices to have the same vertex as successor;
these successor vertices with in-degree \( > 1 \) are \emph{shared}.
Two examples are denoted in Figures~\ref{fig:termgraph:shared1} and~\ref{fig:termgraph:shared2}.

Each term has a natural representation as a tree.
Formally, for a term \( s \) we let
\( \toGraph{s} = (\Pos(s),\gLab,\gSucc,\emptyPos) \) where
\( \gLab(p) = \gApp \) if \( s|_p = s_1 s_2 \) and
\( \gLab(p) = \afun \) if \( s|_p = \afun \);
\( \gLab(p) \) is not defined if \( s|_p \) is a variable;
and \( \gSucc(p) = (1 \cdot p)(2 \cdot p) \) if \( s|_p = s_1 \app s_2 \) and
\( \gSucc(p) = \gEmptyW \) otherwise.
Essentially, \( \toGraph{s} \) maintains the positioning structure of \( s \)
and forgets variable names.
For example, Figure~\ref{fig:termgraph:simplified} denotes both
\( \toGraph{\add \app x \app y} \) and \( \toGraph{\add \app x \app x} \).

Our next step is to \emph{reduce} term graphs using rules.  We limit interest to
\emph{left-linear} rules, which includes all rules in \( \rulesExtd \) (as
\( \rules \) is orthogonal, and the rules in \( \rules_f \) are ground).
To define reduction, we will need some helper definitions.

\begin{definition}
Let \( \aGraph = (\gVert,\gLab,\gSucc,\gRoot) \) and \( v \in \gVert \).
The \defName{subgraph} \( \restr{\aGraph}{v} \) of \( \aGraph \)
with root \( v \) is the term graph \( (\gVert',\gLab',\gSucc',v) \) where \( \gVert' \) contains those nodes \( v' \in \gVert \) such that there exists a path from \( v \) to \( v' \) and \( \gLab', \gSucc' \) are respectively the restrictions of \( \gLab \)
and \( \gSucc \) to \( \gVert' \).
\end{definition}

\begin{definition}\label{def:graph-homomorphism}
A \defName{homomorphism} between two term graphs
\( \aGraph \) and \( \bGraph \)
is a function \( \gHom : \gVert_\aGraph \arrfunc \gVert_\bGraph \)
with \( \gHom(\gRoot_{\aGraph}) = \gRoot_{\bGraph} \),
and for \( \aVert \in \gVert_\aGraph \) such that \( \gLab_\aGraph(v) \)
is defined,
\( \gLab_\bGraph(\gHom(\aVert)) = \gLab_\aGraph(\aVert) \) and
\( \gSucc_\bGraph(\gHom(\aVert)) = \gHom(\aVert_1) \ldots \gHom(\aVert_k) \)
when \( \gSucc_\aGraph(\aVert) = \aVert_1 \ldots \aVert_k \).
(If \( \gLab_\aGraph(\aVert) \) is undefined,
\( \gLab_\bGraph(\gHom(\aVert)) \) and \(\gSucc_\bGraph(\gHom(\aVert)) \)
may be anything.)
\end{definition}

\begin{definition}
  A \defName{redex} in \( \aGraph \) is a triple \( (\rho,\aVert,\gHom) \)
  consisting of some rule \( \rho = \ell \arrz r \in \rulesExtd \),
  a vertex \( \aVert \) in \( \gVert_\aGraph \),
  and a homomorphism \( \gHom : \toGraph{\ell} \arrfunc \restr{\aGraph}{\aVert} \).
\end{definition}

\begin{lemma}\label{lem:novarpath}
  Let \( (\ell \arrz r,\aVert,\gHom) \) be a redex in \( \aGraph \),
  \( x \in \vars{\ell} \), and let \( a_x \) be the corresponding vertex in
  \( \toGraph{\ell} \).
  Then \( \aVert \neq \gHom(a_x) \), nor is there a path from \( \gHom(a_x) \)
  to \( \aVert \).
\end{lemma}

\begin{proof}
Note that \( \ell \) is not a variable, so there is a (non-empty) path from
the root of \( \toGraph{\ell} \) to each \( a_x \); hence, by definition of
homomorphism, there is a path from \( \gHom(\gRoot_{\toGraph{\ell}}) =
\aVert \) to each \( \gHom(a_x) \).  Since \( \aGraph \) is acyclic, this means
that \( \gHom(a_x) \) cannot be \( \aVert \), nor can there be a path from
\( \gHom(a_x) \) back to \( \aVert \).
\end{proof}

\begin{definition}
    Let \( \aGraph \) be a term graph and \( \aVert_1, \aVert_2 \)
    vertices in \( \aGraph \) such that no path exists from \( \aVert_2 \) to
    \( \aVert_1 \).
    The \defName{redirection} of \( \aVert_1 \) to \( \aVert_2 \) is the term
    graph \( \aGraph[\aVert_1 \gRedirect \aVert_2] \equiv
    (\gVert_\aGraph, \gLab_\aGraph, \gSucc_{\aGraph}', \gRoot_\aGraph') \)
    with \( \gSucc_{\aGraph}' \) and \( \gRoot_\aGraph' \) given by:
    \[
    {\gSucc_{\aGraph}'(\aVert)}_i =
    \begin{cases*}
        \aVert_2, & if \( {\gSucc_\aGraph(\aVert)}_i = \aVert_1 \) \\
        {\gSucc_\aGraph(\aVert)}_i, & otherwise
    \end{cases*}
    \quad
    \gRoot_\aGraph' =
    \begin{cases*}
    v_2 & if \( \gRoot_\aGraph = v_1 \) \\
    \gRoot_\aGraph & otherwise \\
    \end{cases*}
    \]
\end{definition}
That is, we replace every reference to \( \aVert_1 \) by a reference to \( \aVert_2 \),
which does not introduce any cycles because there is no path from \( \aVert_2 \) to
\( \aVert_1 \).
With these definitions in hand, we can define \emph{contraction} of term graphs:

\begin{definition}\label{bff:def:graph-rewriting-contraction}
  Let \( \aGraph \) be a term graph, and \( (\rho,\aVert,\gHom) \) a redex in
  \( \aGraph \) with \( \rho \in \rulesExtd \), such that no other vertex \( v' \)
  in \( \restr{\aGraph}{\aVert} \) admits a redex (so \( v \) is an
  \emph{innermost redex position}).
  Denote \( a_x \) for the position of variable \( x \) in \( \ell \), and
  recall that \( a_x \) is a vertex in \( \toGraph{\ell} \).  By left-linearity,
  \( a_x \) is unique for \( x \in \vars{\ell} \).
  The \defName{contraction} of \( (\rho, \aVert, \gHom) \) in \( \aGraph \) is
  the term graph \( \dGraph \) produced after the following steps:
  \( \bGraph \) (building),
  \( \cGraph \) (redirection), and
  \( \dGraph \) (garbage collection).

  \begin{description}
  \item[(building)] Let \( \bGraph = (\gVert_\bGraph,\gLab_\bGraph,
    \gSucc_\bGraph,\gRoot_\aGraph) \) where:
    \begin{itemize}
    \item \( \gVert_\bGraph = \gVert_\aGraph \uplus \{ \overline{p} \in \Pos(r)
      \mid r|_p\ \text{is not a variable} \} \) (\( \uplus \) means
      disjoint union);
    \item for \( \aVert \in \gVert_\aGraph \): \( \gLab_\bGraph(\aVert) =
      \gLab_\aGraph(\aVert) \) and \( \gSucc_\bGraph(\aVert) = \gSucc_\aGraph(
      \aVert) \)
    \item for \( p \in \gVert_\bGraph \) with \( r|_p \) not a variable:
      \begin{itemize}
      \item \( \gLab_\bGraph(\overline{p}) = \afun \) if \( r|_p = \afun \) and
        \( \gLab_\bGraph(\overline{p}) = \gApp \) otherwise
      \item \( \gSucc_\bGraph(\overline{p}) = \gEmptyW \) if
        \( r|_p = \afun \); otherwise, \( \gSucc_\bGraph(\overline{p}) =
        \psi(1 \cdot p) \psi(2 \cdot p) \) \\
        Here, \( \psi(q) = \overline{q} \) if \( r|_q \) is not a variable; if
        \( r|_q = x \) then \( \psi(q) = \gHom(a_x) \).
      \end{itemize}
    \end{itemize}
  \item[(redirection)] If \( r \) is a variable \( x \)
    (so \( \bGraph = \aGraph \)), then let
    \( \cGraph = \aGraph[\aVert \gRedirect \gHom(a_x)] \).
    Otherwise, let \( \cGraph = \bGraph[\aVert \gRedirect \overline{\emptyPos}] \),
    so with all references to \( \aVert \) redirected to the root vertex for
    \( r \).
    (Note that in the first case, there is no path from \( \gHom(a_x) \) to
    \( \aVert \) by Lemma~\ref{lem:novarpath}; and in the second case, every path
    from \( \overline{\emptyPos} \) into \( \gVert_\aGraph \) passes through some
    \( \gHom(a_y) \), so again by Lemma~\ref{lem:novarpath} there is no path from
    \( \overline{\emptyPos} \) to \( \aVert \), and therefore the redirection is
    allowed.)
  \item[(garbage collection)] Let \( \dGraph := \restr{\cGraph}{\gRoot_\cGraph}
    \) (so remove unreachable vertices).
  \end{description}
We then write \( \aGraph \arrg \dGraph \) in one step,
and \( \aGraph \arrg^n \dGraph \) for the n-step reduction.
\end{definition}

We illustrate this with two examples.
First, we aim to rewrite the graph of
\cref{fig:reducibleterm} with a rule \( \add \app \zero \app y \arrz y \) at
vertex \( \aVert \).
Since the right-hand side is a variable, the building phase does nothing.
The result of the redirection phase is given in
\cref{fig:varredirect}, and the result of the garbage collection in
\cref{fig:vargarbage}.

\begin{figure}[H]
\centering
\hfill
\begin{subfigure}[b]{0.3\textwidth}
  \centering
  \begin{tikzpicture}[node distance=0.3cm]
    \node[root] at (0,0) (e) { \( \gApp \) };
    \node (1)   [below left=of e] {\( \suc \)};
    \node (2)   [below right=of e] {\(\aVert\): \( \gApp \)};
    \node (21)  [below left=of 2] {\( \gApp \)};
    \node (22)  [below right=of 2] {\( \gApp \)};
    \node (211) [below left=of 21] {\( \add \)};
    \node (212) [below right=of 21] {\( \zero \)};
    \node (221) [below left=of 22] {\( \suc \)};
    \draw [->] (e) -- (1);
    \draw [->] (e) -- (2);
    \draw [->] (2) -- (21);
    \draw [->] (2) -- (22);
    \draw [->] (21) -- (211);
    \draw [->] (21) -- (212);
    \draw [->] (22) -- (221);
    \draw [->] (22) to [out=-45,in=45,looseness=1] (212);
  \end{tikzpicture}
\caption{}\label{fig:reducibleterm}
\end{subfigure}
\hfill
\begin{subfigure}[b]{0.3\textwidth}
  \centering
  \begin{tikzpicture}[node distance=0.3cm]
    \node[root] at (0,0) (e) { \( \gApp \) };
    \node (1) [below left=of e] {\( \suc \)};
    \node (2) [below right=of e] {\(\aVert\): \( \gApp \)};
    \node (21) [below left=of 2] {\( \gApp \)};
    \node (22) [below right=of 2] {\( \gApp \)};
    \node (211) [below left=of 21] {\( \add \)};
    \node (212) [below right=of 21] {\( \zero \)};
    \node (221) [below left=of 22] {\( \suc \)};
    \draw [->] (e) -- (1);
    \draw [->] (2) -- (21);
    \draw [->] (2) -- (22);
    \draw [->] (21) -- (211);
    \draw [->] (21) -- (212);
    \draw [->] (22) -- (221);
    \draw [->] (22) to [out=-45,in=45,looseness=1] (212);
    \draw [redirect arrow] (e) to [out=-45,in=90,looseness=1] (22);
  \end{tikzpicture}
\caption{}\label{fig:varredirect}
\end{subfigure}
\hfill
\begin{subfigure}[b]{0.3\textwidth}
  \centering
  \begin{tikzpicture}[node distance=0.3cm]
    \node[root] at (0,0) (e) { \( \gApp \) };
    \node (1) [below left=of e] {\( \suc \)};
    \node (22) [below right=of e] {\( \gApp \)};
    \node (221) [below left=of 22] {\( \suc \)};
    \node (222) [below right=of 22] {\( \zero \)};
    \node (hidden) [below of=222, node distance=0.8cm] {};
    \draw [->] (e) -- (1);
    \draw [->] (e) -- (22);
    \draw [->] (22) -- (221);
    \draw [->] (22) -- (222);
  \end{tikzpicture}
\caption{}\label{fig:vargarbage}
\end{subfigure}
\caption{Reducing a graph with the rule \( \add \app \zero \app y \arrz y \)}
\end{figure}

Second, we consider a reduction of the graph in \cref{fig:multbase} by
\( \mult \app (\suc \app x) \app y \arrz \add \app y \app (\mult \app x \app y) \).
Unlike the previous example, this graph has sharing.
\cref{fig:multbuild} shows the
result of the building phase, with the vertices and edges added during this phase in red.
Redirection sets the root to the squared node
(the root of the right-hand side),
and the result after garbage collection is in \cref{fig:multcomplete}.
\begin{figure}[H]
  \centering
  \begin{subfigure}[b]{0.23\textwidth}
    \centering
      \begin{tikzpicture}[node distance=0.3cm]
          \node[root] at (0,0) (0) {\( \gApp \)};
          \node (1) [below left=of 0] { \( \gApp\)};
          \node (11) [below left=of 1 ] {\( \mult \)};
          \node (y) [below of=0, node distance=2cm] {\( \gApp \)};
          \node (121) [below left=of y] {\( \suc \)};
          \node (x) [below right=of y] {\( \zero \)};

          \draw[->] (0) -- (1);
          \draw[->] (0) -- (y);
          \draw[->] (1) -- (11);
          \draw[->] (1) -- (y);
          \draw[->] (y) -- (121);
          \draw[->] (y) -- (x);
      \end{tikzpicture}
  \caption{}\label{fig:multbase}
  \end{subfigure}
  \begin{subfigure}[b]{0.43\textwidth}
    \centering
      \begin{tikzpicture}[node distance=0.3cm]
          \node[root] at (0,0) (0) {\( \gApp \)};
          \node (1) [below left=of 0] { \( \gApp\)};
          \node (11) [below left=of 1 ] {\( \mult \)};
          \node (y) [below of=0, node distance=2cm] {\( \gApp \)};
          \node (121) [below left=of y] {\( \suc \)};
          \node (x) [below right=of y] {\( \zero \)};

          \draw[->] (0) -- (1);
          \draw[->] (0) -- (y);
          \draw[->] (1) -- (11);
          \draw[->] (1) -- (y);
          \draw[->] (y) -- (121);
          \draw[->] (y) -- (x);

          \node[rhsRoot,red] (0') at ([shift={(2.3cm,0.5cm)}]0) {\gApp};
          \node (1') [red,below left=of 0'] { \( \gApp\)};
          \node (11') [red,below left=of 1'] {\( \add \)};
          \node (2')  [red,below right=of 0'] {\( \gApp \)};
          \node (21') [red, below left=of 2' ] {\( \gApp \)};
          \node (211') [red,below left=of 21' ] {\( \mult \)};

          \draw[red,->] (0') -- (1');
          \draw[red,->] (0') -- (2');
          \draw[red,->] (1') -- (11');
          \draw[red,->] (1') to [out=-45,in=70,looseness=1] (y);
          \draw[red,->] (2') -- (21');
          \draw[red,->] (2') to [out=-45, in=20, looseness=1] (y);
          \draw[red,->] (21') -- (211');
          \draw[red,->] (21') to [out=-45,in=30,looseness=1] (x);
      \end{tikzpicture}
  \caption{}\label{fig:multbuild}
  \end{subfigure}
  \begin{subfigure}[b]{0.3\textwidth}
    \centering
      \begin{tikzpicture}[node distance=0.3cm]
          \node[root] (0') at ([shift={(3.5cm,0cm)}]0) {\gApp};
          \node (1') [below left=of 0'] { \( \gApp\)};
          \node (11') [below left=of 1'] {\( \add \)};
          \node (2')  [below right=of 0'] {\( \gApp \)};
          \node (21') [below left=of 2' ] {\( \gApp \)};
          \node (211') [below =of 21' ] {\( \mult \)};
          \node (y) [below of =1', node distance=1.8cm] {\( \gApp \)};
          \node (121) [below left=of y] {\( \suc \)};
          \node (x) [below right=of y] {\( \zero \)};

          \draw[->] (y) -- (121);
          \draw[->] (y) -- (x);

          \draw[->] (0') -- (1');
          \draw[->] (0') -- (2');
          \draw[->] (1') -- (11');
          \draw[->] (1') -- (y);
          \draw[->] (2') -- (21');
          \draw[->] (2') to [out=-45, in=45, looseness=1] (y);
          \draw[->] (21') -- (211');
          \draw[->] (21') to [out=-45,in=30,looseness=1] (x);
      \end{tikzpicture}
  \caption{}\label{fig:multcomplete}
  \end{subfigure}
  \hfill
  \caption{Reducing a term graph with substantial sharing}
\end{figure}

Note that, even when a term graph \( \aGraph \) is not a tree, we can find a
corresponding term: we assign a variable \( \mathit{var}(\aVert) \) to each
unlabeled vertex \( \aVert \) in \( \aGraph \), and let:
\[
  \theta(\aVert) =
  \left\{
  \begin{array}{ll}
  \theta(\aVert_1) \app \theta(\aVert_2) & \text{if}\
    \gLab(\aVert) = \gApp\ \text{and}\ \gSucc(\aVert) = \aVert_1\aVert_2 \\
  \afun & \text{if}\ \gLab(\aVert) = \afun \\
  \mathit{var}(\aVert) & \text{if}\ \gLab(\aVert)\ \text{is undefined} \\
  \end{array}
\right.
\]
Then we may define \( \fromGraph{\aGraph} = \theta(\gRoot_\aGraph) \).
For a linear term (that is, one in which no variable occurs more than once)
, clearly \( \fromGraph{\toGraph{s}} = s \) modulo variable renaming.
We make the following observation:

\begin{lemma}\label{lem:simulation}
Assume given a term graph \( \aGraph \) such that there is a path from
\( \gRoot_\aGraph \) to every vertex in \( \gVert_\aGraph \), and let
\( \fromGraph{\aGraph} = s \).
If \( \aGraph \arrg \bGraph \) then \( \fromGraph{\aGraph} \arrz_{\rulesExtd}^+
\fromGraph{\bGraph} \).
Moreover, if \( s \arrz_{\rulesExtd} t \) for some \( t \),
then there exists \( \bGraph \) such that \( \aGraph \arrg \bGraph \).
\end{lemma}

Consequently,
if \( \arrz_{\rulesExtd} \) is terminating, then so is \( \arrg \);
and if \( \toGraph{s} \arrg^n \aGraph \) for some ground term \( s \)
then \( s \arrz_{\rulesExtd}^* \fromGraph{\aGraph} \) in at least \( n \) steps.
Notice that
if \( \aGraph \) does not admit any redex,
then \( \fromGraph{\aGraph} \) is in normal form.
Moreover,
since
\( \rulesExtd = \rules \cup \rules_f \) is orthogonal
(as \( \rules \) is orthogonal and the \( \rules_f \) rules are non-overlapping)
and therefore confluent,
this is the \emph{unique} normal form of \( s \).
We conclude:

\begin{corollary}\label{cor:simulation}
If \( \toGraph{\funF \app \sfSymbol \app \encode{w}} \arrg^n \aGraph \), then
\( n \leq D(|f|,|w|) \);
and if \( \aGraph \) cannot be reduced by \( \arrg \),
then \( \fromGraph{\aGraph} = \encode{\aTTFunc(f,w)} \).
\end{corollary}

\subsection{Bringing Everything Together}

We are now ready to complete the soundness proof following the recipe at the
start of the section.  Towards the third bullet point, we make the following
observation.

\begin{lemma}\label{lem:increaseconstant}
There is a constant \( a \) such that, whenever \( \aGraph \arrg \bGraph \)
by a rule in \( \rules \), then \( |\bGraph| \leq |\aGraph| + a \), where
\( |\aGraph| \) denotes the total number of nodes in the graph \( \aGraph \).
\end{lemma}

\begin{proof}
In a step using a rule \( \ell \arrz r \), the number of nodes in the graph
can be increased at most by \( |\toGraph{r}| \).  As there are only finitely
many rules in \( \rules \), we can let \( a \) be the number of nodes in the
largest graph for a right-hand side \( r \).
\end{proof}

To see that graph rewriting with \( \funS_f \) can be implemented in an
efficient way, we observe that the size of any intermediate graph in the
reduction \( \toGraph{\funG \app \encode{w}} \arrzT \toGraph{q} \) is
polynomially bounded by a second-order polynomial over \( |f|, |w| \):

\begin{lemma}\label{lem:termsizebounded}
There is a second-order polynomial \( Q \) such that
if
\( \toGraph{\funF \app \sfSymbol \app \encode{w}} \arrg^* \bGraph \),
then \( |\bGraph| \leq Q(|f|, |w|) \).
\end{lemma}

\begin{proof}
Let \( Q(F,x) := x + D(F,x) * (a + F(B(F, x))) \), where
\( D \) is the polynomial from Lemma~\ref{lemma:rules-extd-compatible},
\( a \) is the constant from Lemma~\ref{lem:increaseconstant}, and
\( B \) is the polynomial from Section~\ref{subsec:oraclebound}.
This suffices, because there are at most \( D(|f|,|w|) \) steps
(Lemma~\ref{lemma:rules-extd-compatible}, Corollary~\ref{cor:simulation}),
each of which increases the graph size by at most \( \max(a,|f|(B(|f|,|w|))) \)
according to \cref{lem:increaseconstant,lemma:oracle-subterm-lemma}.
(Note that, by the Oracle Subterm lemma, \( B \) takes the size \( |w| \) of the
original input \( w \) as its second argument, \emph{not} the current size of
the term.)
\end{proof}

All in all, we are finally ready to prove the \emph{soundness} side of the
main theorem:

\begin{theorem}[Soundness]\label{thm:soundness}
  Let \( \rules \) be a finite orthogonal STRS admitting a
  polynomially bounded interpretation.
  If \( \funF \) computes a type-2 functional \( \aTTFunc \),
  then \( \aTTFunc \in \bff \).
\end{theorem}

\begin{proof}
  Given \( \trs \), we can construct an OTM \( \aOTM \) so that for a given
  \( \aTOFunc \in \bffDom \arrfunc \bffDom \), the machine \( \aOTM_f \) executed
  on \( w \in \bffDom \) computes the normal form of
  \( \funF \app \sfSymbol \app \encode{w} \) under \( \arrz_{\rulesExtd} \)
  using graph rewriting.
  We omit the details of the construction, but observe:
  \begin{itemize}
  \item
    that we can represent each graph in polynomial space in the size of the
    graph; simply, we can make use of an adjacency-matrix representation of
    the graph;
  \item
    that we simulate any rewriting step that does not call the oracle (so using a
    rule in \( \rules \)) following the contraction
    algorithm we defined in Definition~\ref{bff:def:graph-rewriting-contraction}, which
    is clearly feasible to do in polynomial time in the size of the graph; of course
    this could in principle involve the creation of new nodes and edges, but the amount of them is statically bounded by the size of the underlying set or rewrite rules (excluding those involving the oracle).
  \item
    and that each oracle call (implemented in rewriting by a
    \( \rules_\aTOFunc \)-step \( \funS_\aTOFunc \app \encode{x}
    \arrz \encode{y} \)) is resolved by copying \( \encode{x} \) to the query
    tape, transitioning to the query state, and from the answer state copying
    \( \encode{y} \) from the answer tape to the main tape.
    By Lemma~\ref{lemma:oracle-subterm-lemma} this is doable in polynomial time
    in \( |f|, |w| \) and the graph size.
  \end{itemize}
  By Lemma~\ref{lem:termsizebounded}, graph sizes are bounded by a
  polynomial over \( |f|, |w| \), so using the above reasoning, the same
  holds for the cost of each reduction step.
  In summary: the total cost of \( \aOTM_\aTOFunc \) running on \( w \) is
  bounded by a second-order polynomial in terms of \( |\aTOFunc| \) and \( |w| \).
  As \( \aOTM_\aTOFunc \) simulates \( \rulesExtd \) via graph rewriting and
  \( \rulesExtd \) computes \( \aTTFunc \), \( \aOTM \) also computes \( \aTTFunc \).
  By Definition~\ref{def:bff-class}, \( \aTTFunc \) is in \( \bffT \).
\end{proof}

\section{Completeness}\label{sec:completeness}

Recall from Section~\ref{sec:rw-bff} that
to prove completeness we have to show the following:
if a given \mbox{type-2} functional \( \aTTFunc \) is in \( \bffT \),
then there exists an orthogonal STRS computing \( \aTTFunc \)
and admitting a polynomially bounded interpretation.
In this section,
we prove this implication by providing an encoding of polynomial time Oracle Turing Machines
as STRSs that admit a polynomially bounded interpretation.

The encoding is divided into three steps.
In Section~\ref{sec:completeness:constructors}, we define the function symbols
that allow us to encode any possible machine configuration as terms.
In Section~\ref{sec:completeness:steps}, we encode machine transitions as reduction
rules that rewrite configuration terms.
This allows us to simulate one transition in an OTM as one or more rewriting steps
on the corresponding rewrite system.
Lastly, we design an STRS which simulates a complete execution of an OTM in
polynomially many steps.
Achieving this polynomial bound is non-trivial and is done in
Section~\ref{sec:completeness:bound} and Section~\ref{sec:completeness:execution}.

Henceforth, we assume given a fixed OTM \( \aOTM \), and a second-order
polynomial \( \polM \), such that \( \aOTM \) operates in time \( \polM \).
For simplicity, we assume the machine has only three tapes (one input/output
tape, one query tape, one answer tape); that each non-oracle transition only
operates on one tape (i.e., reading/writing and moving the tape head); and
that we only have tape symbols \( \{ 0,1,\blank \} \).

\subsection{Encoding Machine Configurations as Terms}\label{sec:completeness:constructors}

Recall from \cref{sec:rw-bff} that we have \( \bito,\biti \hasType \bit,\
\consInfix \hasType \bit \arrtype \word \arrtype \word \) and \( \nil \hasType
\word \), which are the basic constructors for encoding binary words as terms.
To represent a (partial) tape, we also introduce \( \bitb \hasType
\bit \) for the blank symbol.
Now for instance a tape with content
\( 011\blank 01\blank\blank\cdots \) (followed by infinitely
many blanks) may be represented as the list
\( \makelist{\bito\listSep\biti\listSep\biti\listSep\bitb\listSep\bito
\listSep\biti} \) of type \( \word \).
We may also add an arbitrary number of blanks at the end of the
representation; e.g.,
\( \makelist{\bito\listSep\biti\listSep\biti\listSep\bitb\listSep\bito
\listSep\biti\listSep\bitb\listSep\bitb} \).

We can think of a \emph{tape configuration} --- the combination of a tape
and the position of the tape head --- as a finite word
\( w_1 \ldots w_{p - 1} \cfgSep w_p w_{p + 1} \ldots w_k \) (followed by
infinitely many blanks). Here, the tape's head is reading the symbol
\( w_p \).  We can split this tape into two
components: the \emph{left} word \( w_1 \ldots w_{p - 1} \), and the
\emph{right} word \( w_p \ldots w_k \).
To represent a tape configuration as a term, we introduce three symbols:
\begin{align*}
    \tapeL &\hasType \word \arrtype \leftTy
    &\tapeR &\hasType \word \arrtype \rightTy
    &\tapeSplit &\hasType \leftTy \arrtype \rightTy \arrtype \tapeTy
\end{align*}
Here, \( \tapeL, \tapeR \) hold the content%
\footnote{While we technically do not need these two constructors (we could have
\( \tapeSplit \hasType \word \arrtype \word \arrtype \tapeTy \)), they serve
to make configurations more human-readable.
}
of the left and right
split of the tape, respectively.
For convenience in rewriting transitions, later on, we encode the
left side of the split in reverse order.
More precisely, we encode the configuration
\( w_1 \ldots w_{p - 1} \cfgSep w_p w_{p + 1} \ldots w_k \) as the term
\[ \tapeSplit \app (\tapeL \app \makelist{w_{p-1}\listSep\dots\listSep w_2
\listSep w_1}) \app (\tapeR \app \makelist{w_p\listSep\dots\listSep
w_{k-1}\listSep w_k}) \]
The symbol currently being read is the first element of the list below
\( \tapeR \); in case of \( \tapeR \app \nil \), this symbol is \( \blank \).
For a concrete example, a tape configuration \( 1\blank 0\cfgSep 10 \) is
represented by:
\( \tapeSplit \app (\tapeL \app \makelist{\bito\listSep\bitb\listSep\biti})
\app (\tapeR \app \makelist{\biti\listSep\bito}) \).
Since we have assumed an OTM with three tapes, a configuration of the machine
at any moment is a tuple \( (\aState, t_1,t_2,t_3) \), with \( \aState \) a state
and \( t_1,t_2,t_3 \) tape configurations.  To represent machine configurations,
we introduce, for each state \( \aState \), a symbol \( \stateSymb \hasType
\tapeTy \arrtype \tapeTy \arrtype \tapeTy \arrtype \cfgTy \).  Thus, a
configuration \( (\aState,t_1,t_2,t_3) \) is represented by a term
\( \stateSymb \app T_1 \app T_2 \app T_3 \). %

\begin{example}\label{ex:initialconfig}
The initial configuration for a machine \( \aOTM_f \) on input \( w \) is a
tuple of the form
\( (\startState,\cfgSep w,\cfgSep\blank,\cfgSep\blank) \).  This is
represented by the term
\[ \init{w} := \startStateSymb \app
  (\tapeSplit \app (\tapeL \app \nil) \app (\tapeR \app \encode{w})) \app
  (\tapeSplit \app (\tapeL \app \nil) \app (\tapeR \app \nil)) \app
  (\tapeSplit \app (\tapeL \app \nil) \app (\tapeR \app \nil)) \]
\end{example}

To interpret the symbols from this section, we let
\( (\sizeInt{\asort},\sizeGe_{\asort}) := (\Nat,\geq) \) for all \( \asort \),
let \( \funcinterpretcost{\afun} = \fatlambda x_1 \dots x_m.0 \) whenever
\( \afun \) takes \( m \) arguments, and for the sizes:
\[
\begin{array}{rclcrclcrclcrl}
\funcinterpretsize{\bito} & = & 0 & \quad &
  \funcinterpretsize{\nil} & = & 0 & \quad &
  \funcinterpretsize{\consInfix} & = & \fatlambda xy.x+y+1 \\
\funcinterpretsize{\biti} & = & 0 & \quad &
  \funcinterpretsize{\tapeL} & = & \fatlambda x.x & \quad &
  \funcinterpretsize{\tapeSplit} & = & \fatlambda xy.x+y \\
\funcinterpretsize{\bitb} & = & 0 & \quad &
  \funcinterpretsize{\tapeR} & = & \fatlambda x.x & \quad &
  \funcinterpretsize{\aStateSymb} & = & \fatlambda xyz.x+y &
  \quad \text{(for all states \( q \))} \\
\end{array}
\]
Hence, we satisfy the size component of the requirements in
\cref{def:polynomially-bounded-int}.
We have \( \sizeinterpret{\encode{w}} = |w| \);
the size of a tape configuration \( w_1\ldots w_{p-1}\cfgSep w_p\ldots w_k \) is
\( k \), and the size of a configuration is the size of its
first and second tapes combined.  We do \emph{not} include the third tape, as
it does not directly affect either the result yielded by the final configuration
(this is read from the first tape), nor the size of a word the oracle
\( f \) is applied on.

\subsection{Encoding Machine Transitions as Rules}\label{sec:completeness:steps}

\newcommand{\transition}[6]{#1 \xRightarrow[#3]{#2/#4,\ #5} #6}

A single step in an OTM can either be an oracle call (a transition from the
\( \mathtt{query} \) state to the \( \mathtt{answer} \) state), or a traditional step: we
assume that an OTM \( \aOTM \) has a fixed set \( \transitions \) of
\emph{transitions} \( \transition{q}{r}{t}{i}{d}{l} \) where
\( q \) is the \emph{input state},
\( l \) the \emph{output state},
\( t \in \{1,2,3\} \) the tape considered (recall that we have assumed that a
non-oracle transition only operates on one tape),
\( r,i \in \{ 0, 1, \blank \} \) respectively the symbol being read and the symbol
being written, and
\( d \in \{\leftM,\rightM\} \) the direction for the read head of tape \( t \)
to move.
We will model the computation of \( \aOTM \) as rules that simulate the small
step semantics of the machine.
Let us describe such rules as follows:

\begin{itemize}
  \item To encode a single transition, let
\( \machinestep \hasType (\word \arrtype \word) \arrtype \cfgTy
   \arrtype \cfgTy \).
For any transition of the form \( \transition{q}{r}{1}{i}{\leftM}{l} \) (so
a transition operating on tape 1, and moving left), we introduce a rule
(where we write \( \encode{0} = \bito,\ \encode{1} = \biti,\ \encode{\blank} =
\bitb \)):
\[
    \machinestep \app \aFuncVar \app ( \aStateSymb \app
    (
        \tapeSplit \app
        (\tapeL \app (\avar \consInfix \bvar)) \app
        (\tapeR \app (\encode{r} \consInfix \cvar ))
    ) \app u \app v )
    \arrz
    \bStateSymb \app
    (
        \tapeSplit \app
        (\tapeL \app \bvar) \app
        (\tapeR \app (\avar \consInfix \encode{i} \consInfix \cvar))
    )
    \app u \app v \\
\]
Moreover, for transitions \( \transition{q}{\blank}{1}{w}{\leftM}{l} \) (so
where \( \blank \) is read), we add a rule:
\[
    \machinestep \app \aFuncVar \app ( \aStateSymb \app
    (
        \tapeSplit \app
        (\tapeL \app (\avar \consInfix \bvar)) \app
        (\tapeR \app \nil )
    ) \app u \app v )
    \arrz
    \bStateSymb \app
    (
        \tapeSplit \app
        (\tapeL \app \bvar) \app
        (\tapeR \app (\avar\consInfix\encode{i}\consInfix \nil))
    )
    \app u \app v
\]
These rules respectively handle the steps where a tape configuration is changed
from \( u_1\ldots u_{p-1}u_p \cfgSep r u_{p+2} \ldots u_k \) to
\( u_1\ldots u_{p-1}\cfgSep u_p i u_{p+2} \ldots u_k \), and where a tape
configuration is changed from \( u_1\ldots u_k \cfgSep \) to
\( u_1\ldots \cfgSep u_k i \).

\item Transitions where \( d = \rightM \), or on the other two tapes, are encoded
similarly.

\item
Next, we encode oracle calls.
Recall that to query the machine for the value
of \( f \) at \( u \), we write \( u \) on the second tape,
move its head to the leftmost position, and enter the query state. Then,
the content of this tape is erased and the image of \( f \)
over \( u \) is written in the third tape.
Visually, this step is represented as:
\[
    (\querystate, \langle\text{tape}_1\rangle,
      v_1 \ldots v_p\cfgSep \encode{u} \blank \ldots,
      \langle\text{tape}_3\rangle)
    \cfgTrans
    (\answerstate, \langle\text{tape}_1\rangle,
      \cfgSep \blank,
      \cfgSep \encode{f(u)} )
\]
This is implemented by the following rules:
\[
  \begin{array}{rcl}
  \machinestep \app \aFuncVar \app ( \querystate \app
    t_1 \app (\tapeSplit \app x \app (\tapeR \app y)) \app t_3 ) &
  \arrz &
  \answerstate \app t_1 \app
    (\tapeSplit \app (\tapeL \app \nil) \app (\tapeR \app \nil))
    \phantom{ABCDEFG} \\
\multicolumn{3}{r}{
    (\tapeSplit \app (\tapeL \app \nil) \app
    (\tapeR \app (F \app (\clean \app y))))}
\end{array}
\]
\[
\begin{array}{rclcrclcrcl}
\clean \app (\bito \consInfix \avar) & \arrz & \bito \consInfix (\clean \app \avar)
& \quad &
\clean \app (\bitb \consInfix \avar) & \arrz & \nil \\
\clean \app (\biti \consInfix \avar) & \arrz & \biti \consInfix (\clean \app \avar)
& \quad &
\clean \app \nil & \arrz & \nil \\
\end{array}
\]
Here, \( \clean \hasType \word \arrtype \word \) turns a word that may have
blanks in it into a bitstring, by reading until the next blank; for instance
replacing
\( \makelist{\bito\listSep\biti\listSep\bitb\listSep\biti} \) by
\( \makelist{\bito\listSep\biti} \).
\end{itemize}

The various \( \machinestep \) rules and the \( \clean \) rules defined above are
non-overlapping because we consider \emph{deterministic} OTMs.
Additionally, they are also left-linear.
We can orient such rules as follows:
\[
\begin{array}{rclcrcl}
\funcinterpretsize{\clean} & = & \fatlambda x.x & \quad &
\funcinterpretcost{\clean} & = & \fatlambda x.x+1 \\
\funcinterpretsize{\machinestep} & = & \fatlambda Fx.x + 1 & \quad &
\funcinterpretcost{\machinestep} & = & \fatlambda F^c F^s x. F^c(x) + x + 2 \\
\end{array}
\]
Note that \( \funcinterpretsize{\machinestep} \) is so simple because the size
of a configuration does not include the size of the answer tape.
From these simulation rules, the following result is straightforward.

\newcommand{\fromOTMtoTm}[1]{{[#1]}}

\begin{lemma}\label{lem:stepencode}
Let \( \aOTM_f \) be an OTM and \( C, D \) be machine configurations of
\( \aOTM_f \) such that \( C \cfgTrans D \).
Then \( \machinestep \app \funS_f \app \fromOTMtoTm{C} \arrzT
\fromOTMtoTm{D} \),
where \( \fromOTMtoTm{C} \) is the term encoding of %
\(C\).
\end{lemma}

\subsection{A Bound on the Number of Steps}\label{sec:completeness:bound}

To generalize from performing a single step of the machine to tracing a full
computation on the machine level, the natural idea would be to
define rules such as:
\[
\begin{array}{rcll}
\execute \app F \app (\aStateSymb \app x \app y \app z) & \arrz &
  \execute \app F \app (\machinestep (\aStateSymb \app x \app y \app z)) &
  \text{for}\ \aStateSymb \neq \finalstate \\
\execute \app F \app (\finalstate \app (\tapeSplit \app (\tapeL \app x) \app
  (\tapeR \app w)) \app y \app z) & \arrz & \clean \app w \\
\end{array}
\]
Then, reducing \( \execute \app \sfSymbol \app \init{w} \) to normal form
simulates a full OTM execution of \( \aOTM_f \) on input \( w \).
Unfortunately, this rule does not admit an interpretation, as it may be
non-terminating.
A solution could be to give \( \execute \) an additional argument
\( \unary{N} \)
suggesting an execution in at most \( N \) steps; this argument would ensure
termination, and could be used to find an interpretation.

The challenge, however, is to compute a bound on the number of steps in the OTM:\@
the obvious thought is to compute \( \polM(|f|,|w|) \), but this cannot in
general be done in polynomial time because the STRS does not have access to
\( |f| \): since \( |f|(i) = \max \{ x \in \Nat \mid |x| \leq i \} \), there
are exponentially many choices for \( x \).

To solve this, and following~\cite[Proposition 2.3]{kapron:cook:96}, we observe that
it suffices to know a bound for \( f(x) \) for only those \( x \) on
which the oracle is actually questioned.
That is, for \( A \subseteq \bffDom \), let \( \limitsize{f}{A} = \fatlambda
n. \max \{ |f(x)| \mid x \in A \wedge |x| \leq n \} \).  Then:

\begin{lemma}\label{lem:limitsizebounds}
Suppose an OTM \( \aOTM_f \) runs in time bounded by \( \polM(|f|,|w|) \) on
input \( w \).
If \( \aOTM_f \) transitions in \( N \) steps from its initial state to some
configuration \( C \), calling the oracle only on words in \( A \subseteq
\bffDom \), then \( N \leq \polM(\limitsize{f}{A},|w|) \).
\end{lemma}

\begin{proof}[Proof sketch]
We construct \( f' \) with \( f'(x) = 0 \) if \( x \notin A \) and \( f'(x) =
f(x) \) if \( x \in A \).  Then \( |f'| = \limitsize{f}{A} \), and \( \aOTM_{f'}
\) runs the same on input \( w \) as \( \aOTM_f \) does.
\end{proof}

Now, for \( A \) encoded as a term \( \mathsf{A} \) (using symbols \(\setnil
\hasType \setsort,\ \setcons \hasType \word \arrtype \setsort \arrtype \setsort
\)), we can compute \( \limitsize{f}{A} \) using the rules below, where we use
unary integers as in Example~\ref{ex:add} (\( \zero \hasType \nat, \suc \hasType \nat
\arrtype \nat \)), and defined symbols \(
\len \hasType \word \arrtype \nat,\
\maxsymb \hasType \nat \arrtype \nat \arrtype \nat,\
\limit \hasType \word \arrtype \nat \arrtype \word,\
\returnif \hasType \word \arrtype \nat \arrtype \word \arrtype \word,\
\tryapply \hasType (\word \arrtype \word) \arrtype \word \arrtype \nat \arrtype
  \nat,\
\tryall \hasType (\word \arrtype \word) \arrtype \setsort \arrtype \nat \arrtype
  \nat
\).
By design, \( \returnif \app \encode{x} \app \unary{n} \app \encode{y} \)
reduces to \( \encode{y} \) if \( |x| \leq n \) and to \( \nil \) otherwise;
\( \tryapply \app \sfSymbol \app \encode{x} \app \unary{n} \) reduces to
the unary encoding of \( \limitsize{F}{\{x\}}(n) \) and
\( \tryall \app  \sfSymbol  \app \textsf{a} \app  \unary{n} \) yields
\( \limitsize{F}{A}(n) \).
\[
\begin{array}{rclcrclrcl}
\len \app \nil & \arrz & \zero & \quad &
  \len \app (x \consInfix y) & \arrz & \suc \app (\len \app y) \\
\maxsymb \app \zero \app m & \arrz & m & \quad &
  \maxsymb \app (\suc \app n) \app \zero & \arrz & \suc \app n &
  \maxsymb \app (\suc \app n) \app (\suc \app m) & \arrz &
    \suc \app (\maxsymb \app n \app m) \\
\limit \app \nil \app n & \arrz & \nil & \quad &
  \limit \app (x \consInfix y) \app \zero & \arrz & \nil &
  \limit \app (x\consInfix y) \app (\suc \app n) & \arrz &
    x \consInfix (\limit \app y \app n) \\
\returnif \app \nil \app n \app z & \arrz & z & \quad &
  \returnif \app (x \consInfix y) \app \zero \app z & \arrz & \nil &
  \returnif \app (x \consInfix y) \app (\suc \app n) \app z & \arrz &
    \returnif \app y \app n \app z
\end{array}
\]
\[
\begin{array}{rclcrcl}
& & & & \tryapply \app F \app a \app n & \arrz &
  \len \app (\returnif \app a \app n \app (F \app (\limit \app a \app n))) \\
\tryall \app F \app \setnil \app n & \arrz & \zero & \quad &
\tryall \app F \app (\setcons \app a \app tl) \app n & \arrz &
  \maxsymb \app (\tryapply \app F \app a \app n) \app (\tryall \app F \app tl
  \app n) \\
\end{array}
\]

The \cs{} interpretations of these rules are as follows.
\[
\begin{array}{rclcrcl}
  \funcinterpretsize{\len} & = & \fatlambda x.x & \quad\quad &
  \funcinterpretcost{\len} & = & \fatlambda x.x + 1 \\
  \funcinterpretsize{\maxsymb} & = & \fatlambda nm.\max(n,m) & &
  \funcinterpretcost{\maxsymb} & = & \fatlambda nm.n+1 \\
  \funcinterpretsize{\limit} & = & \fatlambda xn.n & &
  \funcinterpretcost{\limit} & = & \fatlambda xn.n+1 \\
  \funcinterpretsize{\returnif} & = & \fatlambda xnz.z & &
  \funcinterpretcost{\returnif} & = & \fatlambda xnz.n+1 \\
  \end{array}
\]
It is easy to see that the corresponding rules are all oriented.

For \( \tryapply \), note that \( \tryapply \app F \app a \app \unary{n} \)
reduces to \( \unary{|F(a)|} \) if \( |a| \leq n \), and to \( \unary{0} \)
otherwise.  Thus, it indeed returns exactly \( \limitsize{F}{\{a\}}(n) \).
\[
\begin{array}{rclcrcl}
\funcinterpretsize{\tryapply} & = & \fatlambda Fan.F(n) & \quad\quad &
\funcinterpretcost{\tryapply} & = & \fatlambda F^c F^s an.
  F^c(n) + F^s(n) + 2 * n + 4 \\
\end{array}
\]
We easily see that \( \sizeinterpret{\tryapply \app a \app n} =
\sizeinterpret{\len \app (\returnif \app a \app n \app
(F \app (\limit \app a \app n)))} \). As for the cost, note that
\[
\begin{array}{cl}
& \totalcost{\len \app (\returnif \app a \app n \app
(F \app (\limit \app a \app n)))} \\
= &
\costinterpret{\len \app (\returnif \app a \app n \app
(F \app (\limit \app a \app n)))}  +
\costinterpret{\returnif \app a \app n \app
(F \app (\limit \app a \app n))}\ + \\ &
\costinterpret{F \app (\limit \app a \app n} +
\costinterpret{\limit \app a \app n} \\
= & (F^c(n) + 1) + (n + 1) + F^s(n) + (n + 1) = F^c(n) + F^s(n) + 2n + 3 \\
\end{array}
\]
Hence, also the \( \tryapply \) rule is oriented.

To interpret sets and the apply rule, we use:
\[
\begin{array}{rclcrclcrclcrcl}
\funcinterpretsize{\setnil} & = & 0 & \quad\quad &
\funcinterpretcost{\setnil} & = & 0 & \quad\quad &
\funcinterpretsize{\setcons} & = & \fatlambda xy.y+1 & \quad\quad&
\funcinterpretcost{\setcons} & = & \fatlambda xy.0 \\
\funcinterpretsize{\tryall} & = &
  \multicolumn{11}{l}{\fatlambda F a n.F(n)} \\
\funcinterpretcost{\tryall} & = &
  \multicolumn{11}{l}{\fatlambda F^c F^s a n.
  1 + a * (F^c(n) + 2 * F^s(n) + 2 * n + 6)}
\end{array}
\]
To see that the rule is oriented, note:
\[
\begin{array}{rcl}
\sizeinterpret{\tryall \app F \app (\setcons \app a \app tl) \app n}
& = & F^s(n) \\
& = & \max(F^s(n),F^s(n)) \\
& = & \sizeinterpret{\maxsymb \app (\tryapply \app F \app a \app n) \app
  (\tryall \app F \app tl \app n)} \\
\end{array}
\]
and
\[
\begin{array}{ll}
& \costinterpret{\tryall \app F \app (\setcons \app a \app tl) \app n} \\
 = & 1 + (tl + 1) * (F^c(n) + 2 * F^s(n) + 2 * n + 6) \\
 = & 1 + tl * (F^c(n) + 2 * F^s(n) + 2 * n + 6) \\
   & \phantom{ABi} + 1 * (F^c(n) + 2 * F^s(n) + 2 * n + 6) \\
 = & \costinterpret{\tryall \app F \app tl \app n} +
   (F^c(n) + 2 * F^s(n) + 2 * n + 6) \\
 = & \costinterpret{\tryall \app F \app tl \app n} +
   \costinterpret{\tryapply \app F \app a \app n} +
   F^s(n) + 2 \\
 = & \costinterpret{\tryall \app F \app tl \app n} +
     \costinterpret{\tryapply \app F \app a \app n} +
     \costinterpret{\maxsymb \app (\tryapply \app F \app a \app n) \app
      (\tryall \app F \app tl \app n)} + 1 \\
 > & \totalcost{\maxsymb \app (\tryapply \app F \app a \app n) \app
      (\tryall \app F \app tl \app n)} \\
\end{array}
\]

Importantly, the \( \limit \) function ensures that, in \( \tryall \app F \app
n \) we never apply \( F \) to a word \( w \) with \( |w| > n \).
Therefore we can let \( \sizeinterpret{\mathsf{A}} = |A| \), the number of words
in \( A \), and have \( \funcinterpretsize{\tryall} = \fatlambda F a n.F(n) \)
and \( \funcinterpretcost{\tryall} = \fatlambda F^c F^s a n.1+a + F^c(n) +
2 * F^s(n) + 2 * n + 6 \).

Now, for a given second-order polynomial \( P \), fixed \( f,n \), and a term
\( \textsf{A} \) encoding a set \( A \subseteq \bffDom \), we can construct
a term \( \Theta^P_{\sfSymbol;\unary{n};\textsf{A}} \) that computes
\( P(\limitsize{f}{A},n) \) using \( \tryall \) and the functions
\( \add,\mult \) from Example~\ref{ex:add}.
By induction on \( P \), we have
\( \sizeinterpret{\Theta^P_{\sfSymbol;\unary{n};\textsf{A}}} = P(|f|,n) \),
while its cost is bounded by a polynomial
over \( |f|,n,|A| \).

\subsection{Finalising Execution}\label{sec:completeness:execution}

Now, we can define execution in a way that can be bounded by a polynomial
interpretation.
We let \( \execute \hasType (\word \arrtype \word) \arrtype
\nat \arrtype \nnat \arrtype \nat \arrtype \setsort \arrtype \cfgTy \arrtype
\word \) and will define rules to reduce expressions \( \execute \app F \app
n \app m \app z \app a \app c \) where
\begin{itemize}
\item \( F \) is the function to be used in oracle calls.
\item \( n - 1 \) is a bound on the number of steps that can be done before
  the next oracle call (or until the machine completes execution).
\item \( m \) is essentially a natural number that represents the number of
  steps that have been done so far.  We use a new sort \( \nnat \) with
  function symbols \( \natz \hasType \nnat \) and \( \natsuc \hasType \nnat
  \arrtype \nnat \) because we will let \( \sizeInt{\nnat} = (\Nat,\leq) \),
  so ordered in the other direction.  This will be essential to find an
  interpretation for \( \execute \).
\item \( z \) is a unary representation of \( |w| \), where \( w \) is the
  input to the OTM.\@
\item \( c \) is the current configuration.
\end{itemize}

Using helper symbols
\( \funF' \hasType (\word \arrtype \word) \arrtype \nat \arrtype \cfgTy
  \arrtype \word \),
\( \execute' \hasType (\word \arrtype \word) \arrtype \nat \arrtype \nnat
  \arrtype \nat \arrtype \setsort \arrtype \cfgTy \arrtype \word \),
\( \extract \hasType \tapeTy \arrtype \word \) and
\( \minus \hasType \nat \arrtype \nnat \arrtype \nat \), we introduce the
rules:
\[
\begin{array}{l}
\funF \app F \app w \arrz \funF' \app F \app (\len \app w) \app
  (\startStateSymb \app
  (\tapeSplit (\tapeL \app \nil) \app (\tapeR \app w)) \app
  (\tapeSplit (\tapeL \app \nil) \app (\tapeR \app \nil)) \app
  (\tapeSplit (\tapeL \app \nil) \app (\tapeR \app \nil))) \\
\funF' \app F \app z \app c \arrz \execute \app F \app
  \Theta^{\polM+1}_{F;z;\setnil} \app \natz \app z \app \setnil \app c \\
\end{array}
\]
\[
\begin{array}{l}
\execute \app F \app (\suc \app n) \app m \app z \app a \app
  (\aStateSymb \app t_1 \app t_2 \app t_3) \arrz \\
  \hfill
  \execute \app F \app n \app (\natsuc \app m) \app z \app
  (\machinestep \app F \app (\aStateSymb \app t_1 \app t_2 \app t_3))
  \ \text{for}\ \aStateSymb \notin \{ \querystate, \finalstate \} \\
\execute \app F \app (\suc \app n) \app m \app z \app a \app
  (\querystate \app t_1 \app t_2 \app t_3) \arrz \\
  \hfill
  \execute' \app F \app n \app
  (\natsuc \app m) \app z \app (\setcons \app (\extract \app t_2) \app a) \app
  (\querystate \app t_1 \app t_2 \app t_3) \\
\execute' \app F \app n \app m \app z \app a \app c \arrz
  \execute \app F \app (\minus \app \Theta^{\polM+1}_{F;z;a} \app m) \app m
  \app z \app a \app (\machinestep \app F \app c) \\
\execute \app F \app n \app m \app z \app a \app (\finalstate \app t_1 \app
  t_2 \app t_3) \arrz \extract \app t_1 \\
\extract \app (\tapeSplit \app (\tapeL \app x) \app (\tapeR \app y)) \arrz
  \clean \app y \\
\minus \app x \app \natz \arrz x
  \quad\quad
\minus \app \zero \app (\natsuc \app y) \arrz \natz
  \quad\quad
\minus \app (\suc \app x) \app (\natsuc \app y) \arrz \minus \app x \app y
  \\
\end{array}
\]
\noindent
That is, an execution on \( \funF \app \sfSymbol \app \encode{w} \) first computes the length of \( w \) and \( \polM(\limitsize{f}{\emptyset},|w|) \),
and uses these as arguments to \( \execute \).
Each normal transition lowers the number \( n \) of steps we are allowed to do and increases the number
\( m \) of steps we have done.
Each oracle transition updates \( A \), and
either lowers \( n \) by one, or updates it to the new value
\( \polM(\limitsize{f}{A},|w|) - m \), since we have already done \( m \)
steps.
Once we read the final state, the answer is read off the first tape.

For the interpretation, note that the unusual size set of \( \nnat \) allows us
to choose \( \funcinterpretsize{\minus} = \fatlambda xy.\max(x-y,0) \) without
losing monotonicity.  Hence, in every step \( \execute \app F \app n \app
m \app z \app a \app c \), the value
\( \max(\polM(\sizeinterpret{F},\sizeinterpret{z}) + 1 - \sizeinterpret{m},
\sizeinterpret{n}) \) decreases by at least one.
Since \( \sizeinterpret{\Theta^{\polM+1}{F;z; a}} = \polM(\sizeinterpret{F},
\sizeinterpret{z}) \) regardless of \( a \), we can use this component as part
of the interpretation.
So let us provide the interpretations for the \( \nnat \) functions symbols and
the two simple rules for \( \extract \) and \( \minus \), given as follows:
\[
\begin{array}{rclcrcl}
\funcinterpretsize{\natz} & = & 0 & \quad\quad &
\funcinterpretcost{\natz} & = & 0 \\
\funcinterpretsize{\natsuc} & = & \fatlambda x.x+1 & &
\funcinterpretcost{\natsuc} & = & \fatlambda x.0 \\
\funcinterpretsize{\extract} & = & \fatlambda x.x & &
\funcinterpretcost{\extract} & = & \fatlambda x.x + 2 \\
\funcinterpretsize{\minus} & = & \fatlambda xy.\max(x-y,0) \\
\funcinterpretcost{\minus} & = & \fatlambda xy.x \\
\end{array}
\]
These functions are all    monotonic, and their rules are oriented
(as can easily be checked).
\newcommand{\interdecrease}[2]{\theta_{F,z,#1,#2}}
\newcommand{\costpoly}[1]{\mathtt{POLY}_{F,z}[#1]}
\newcommand{\makecalc}[2]{\Theta_{#2}^{#1}}
By induction on the polynomial \( P \),
we can find polynomials \( A_P, B_P \) such that
\[
  \totalcost{\makecalc{P}{F;z;a}} \leq
  \sizeinterpret{a} * A_P(F^c,F^s,\sizeinterpret{z}) +
  B_P(F^c,F^s,\sizeinterpret{z})
\]
assuming \( F \), \( z \) and \( a \) are in
normal form.

To define our remaining interpretation functions, first let:
\begin{itemize}
\item \( \interdecrease{n}{m} := \max(\polM(F^s,z) + 1 - m, n) \)
\item \( \costpoly{x} := x * A_{\polM+1}(F^c,F^s,\sizeinterpret{z}) +
  B_{\polM+1}(F^c,F^s,\sizeinterpret{z}) \), so the polynomial bounding
  \( \totalcost{\makecalc{\polM+1}{F;z;a}} \) if \( \sizeinterpret{a} = x \).
\end{itemize}
Then, we can orient the size interpretations of the rewrite rules by the
following interpretation:
\[
\begin{array}{rcl}
\funcinterpretsize{\funF} & = & \fatlambda F n.n + \polM(F,n) + 1 \\
\funcinterpretsize{\funF'} & = & \fatlambda F z c.c + \polM(F,z) + 1 \\
\funcinterpretsize{\execute} & = & \fatlambda F n m z a c.c +
  \interdecrease{n}{m} \\
\funcinterpretsize{\execute'} & = & \fatlambda F n m z a c.c + 1 +
  \interdecrease{n}{m} \\
\end{array}
\]
And the cost interpretations by:
\[\begin{array}{rcll}
\funcinterpretcost{\funF} & = &
\fatlambda F n.
  (\polM(F^s,n)+1) * \\
  & & \phantom{AB} (8 + 3 * \polM(F^s,n) + 2 * n + F^c(\polM(F^s,n) + n + 1) + \\
  & & \phantom{AB} \costpoly{\polM(F^s,n)+1}) + 6 + 2 * n + \costpoly{0} \\
  \\
\funcinterpretcost{\funF'} & = & \fatlambda F z c.
  (\polM(F^s,z)+1) * \\
  & &\phantom{AB} (8 + 3 * \polM(F^s,z) + 2 * c + F^c(\polM(F^s,z)+1 + c) + \\
  & & \phantom{AB} \costpoly{\polM(F^s,z)+1}) + 4 + c + \costpoly{0} \\
  \\
\funcinterpretcost{\execute} & = & \fatlambda F n m z a c.
  \interdecrease{n}{m} * \\
  & & \phantom{AB} (5 + 2 * (\interdecrease{n}{m} + c) +
  F^c(\interdecrease{n}{m} + c) + \\
  & & \phantom{AB} \costpoly{\interdecrease{n}{m} + a} + \polM(F^s,z)) + 3 + \interdecrease{n}{m} + c \\
  \\
  \funcinterpretcost{\execute'} & = & \fatlambda F n m z a c.
  (\interdecrease{n}{m} + 1) * \\
  & & \phantom{AB} (5 + 2 * (\interdecrease{n}{m} + c + 1) +
  F^c(\interdecrease{n}{m} + c + 1) + \\
  & & \phantom{AB} \costpoly{\interdecrease{n}{m} + a} + \polM(F^s,z)) + 1
\end{array}\]
To see that these interpretations are correct, we first observe:
\[\begin{array}{rcl}
  \interdecrease{\suc \app n}{m} & = & \max(\polM(F^s,z) + 1 - m, n + 1) \\
  & = & \max(\polM(F^s,z) + 1 - (m+1), n) + 1 \\
  & = & \interdecrease{n}{m+1} + 1 \\
\end{array}\]
Observe that this holds since \( \max(a + 1, b + 1) = \max(a, b) + 1 \).
We also have the following:
\[\begin{array}{rcl}
  \interdecrease{n}{m} & = & \max(\polM(F^s,z) + 1 - m, n) \\
  & \geq & \max(\polM(F^s,z) + 1 - m, 0) \\
  & = & \max(\polM(F^s,z) + 1 - m, \max(\polM(F^s, z) + 1 - m), 0) \\
  & = & \interdecrease{\sizeinterpret{\minus \app \makecalc{\polM+1}{F;z;m}}}{m}
\end{array}\]
The inequalities now follow by writing out definitions.

\begin{theorem}\label{thm:completeness}
  If \( \aTTFunc \in \bffT \),
  then there exists a finite orthogonal STRS \( \rules \)
  such that \( \funF \) computes \( \aTTFunc \) in \( \rules \)
  and \( \rules \) admits a polynomially bounded interpretation.
\end{theorem}

\section{Conclusion}\label{sec:conclusion}

This paper gives the first characterization of the class of type-2 basic
feasible functionals by term rewriting based on the interpretation method and in particular
on polynomial \cs{} interpretations.
Surely, it would make sense to try to go beyond types of order 2 and try to characterize classes of basic feasible functionals of order at least 3. Such classes are less well understood than their order-2 counterparts and perhaps an analysis based on tools from rewriting could help to shed some light on their nature (see, e.g., Hugo F\'er\'ee's work on this subject~\cite{Feree}).

\bibliographystyle{alphaurl}
\bibliography{references}

\end{document}